%% file: Palindrome_Recognition_In_the_Streaming_Model.tex
\documentclass[a4paper,UKenglish]{article}
 \usepackage[margin=1.0in]{geometry}

\usepackage{microtype}

\usepackage{amsthm}

\usepackage{amsmath}
\usepackage{mathtools}
\usepackage{xfrac}
\usepackage{xspace}
\usepackage{xargs}
\usepackage{amssymb}
\usepackage{enumerate}
\usepackage{authblk}

\newtheorem{theorem}{Theorem}[section]
\newtheorem{lemma}[theorem]{Lemma}

\newtheorem{corollary}[theorem]{Corollary}
\newtheorem{observation}{Observation}

\newtheorem{definition}{Definition}

\usepackage{paralist}

\DeclarePairedDelimiter{\ceil}{\lceil}{\rceil}
\DeclarePairedDelimiter{\floor}{\lfloor}{\rfloor}

\newcommand*{\QEDB}{\hfill\ensuremath{\square}}%

\newcommand{\kdots}{, \dots, } 
\newcommand{\rom}[1]{\expandafter\@showromancap\romannumeral{#1}}
\newcommand{\algA}{\emph{Algorithm ApproxSqrt}\xspace}
\newcommand{\algsA}{\emph{Algorithm Simple ApproxSqrt}\xspace}
\newcommand{\algB}{\emph{Algorithm Exact}\xspace}
\newcommand{\A}{\emph{ApproxSqrt}\xspace}
\newcommand{\sA}{\emph{Simple ApproxSqrt}\xspace}

\newcommand{\B}{\emph{Exact}\xspace}
\newcommand{\algLog}{\emph{Algorithm ApproxLog}\xspace}\newcommand{\aLog}{\emph{ApproxLog}\xspace}
\newcommand{\mfp}{\emph{Master Fingerprints}\xspace}
\newcommand{\comp}{\emph{Compressed Run}\xspace}
\newcommand{\comps}{\emph{Compressed Runs}\xspace}
\newcommand{\current}{\ensuremath{R_{NF}}-entry\xspace}
\newcommand{\currents}{\ensuremath{R_{NF}}-entries\xspace}
\newcommand{\finished}{\ensuremath{R_F}-entry\xspace}
\newcommand{\finisheds}{\ensuremath{R_F}-entries\xspace}
\newcommand{\single}{\ensuremath{R_S}-entry\xspace}
\newcommand{\singles}{\ensuremath{R_S}-entries\xspace}
\newcommand{\withoutlog}{WLOG.\xspace}

\newcommand{\slidingpairs}{\ensuremath{Fingerprint\ Pairs}\xspace}
\newcommand{\slidingpair}{\ensuremath{Fingerprint\ Pair}\xspace}

\newcommand{\ignore}[1]{}

\newcommand{\problemDefLongest}{\emph{Longest Palindromic Substring Problem}\xspace}
\newcommand{\problemDefAll}{\emph{Palindrome Problem}\xspace}

\title{Palindrome Recognition In The Streaming Model}

\author[1]{Petra Berenbrink}
\author[1,2]{Funda Erg{\"u}n}
\author[1]{Frederik Mallmann-Trenn}
\author[1]{Erfan Sadeqi Azer}
\affil[1]{Simon Fraser University}
\affil[2]{Indiana University}

\begin{document}

\maketitle

\begin{abstract}
A palindrome
is defined as a string which reads forwards the same as backwards, like, for example, the string ``racecar''. 
In the \problemDefAll, one tries to find all \emph{palindromes} in a given string.
In contrast, in the case of the  \problemDefLongest, the goal is to find an arbitrary one of the longest palindromes in the string.\\
In this paper we present three algorithms for finding w.h.p. palindromes in the streaming model  for the the above problems, where
at any point in time we are only allowed to use sublinear space.
We first present a one-pass randomized algorithm that solves the \problemDefAll. It has an \emph{additive} error and uses $O(\sqrt n)$ space. We also give two variants of the  algorithm which solve  related and practical problems.
The second algorithm determines the exact locations of {\em all} longest palindromes using two passes and $O(\sqrt n)$ space.
The third algorithm is a one-pass randomized algorithm,
which solves the \problemDefLongest. It has a \emph{multiplicative} error using only $O(\log(n))$ space.
Moreover, we give an $\Omega(n)$ lower bound for any algorithm finding the longest palindrome with probability $1$.
\end{abstract}
\vspace{-2.9mm}
\section{Introduction}\vspace{-2.3mm}
A palindrome 
is defined as a string which reads forwards the same as backwards, e.g., the string ``racecar''. 
In the \problemDefAll one tries to find all \emph{palindromes} (palindromic substrings) in an input string.
A related problem is the \problemDefLongest in which one tries to find 
any one of the longest palindromes in the input.\\
In this paper we regard the streaming version of both problems, where
the input arrives over time (or, alternatively, is read as a stream)
and the algorithms are allowed space sub linear in the size of the input.
%
Our first contribution is a one-pass randomized algorithm that solves the \problemDefAll w.h.p.. It has an \emph{additive} error and uses $O(\sqrt n$) space.
The second contribution is a two-pass algorithm which determines w.h.p. the exact locations of all longest palindromes. 
It uses the first algorithm as the first pass and uses $O(\sqrt n)$ space.
The third is a one-pass randomized algorithm which solves
the \problemDefLongest w.h.p.. It has a \emph{multiplicative} error using $O(\log(n))$ space.
Moreover, we give an $\Omega(n)$ lower bound for any algorithm finding the longest palindrome with probability $1$.
We also give two variants of the first algorithm which solve other related practical problems.

Palindrome recognition is important in computational biology. 
Palindromic structures can frequently be found in proteins and identifying them gives researchers 
hints about the structure of nucleic acids. For example, in \emph{nucleic acid secondary structure prediction}, one is interested in complementary palindromes which are considered in the appendix.
\newline

\noindent{\bf Related work} 
While palindromes are well-studied,
to the best of our knowledge there are no results for the streaming model.
Manacher \cite{Manacher:1975} presents a linear time online algorithm that reports at any time whether all symbols seen so far form a palindrome.
The authors of \cite{Apostolico:1995} show how to modify this algorithm in order to find all palindromic substrings in linear time (using a parallel algorithm).\\
%
Some of the techniques used in this paper have their origin in the streaming pattern matching literature. 
In the \emph{Pattern Matching Problem}, one tries to find all occurrences of a given pattern $P$ in a text $T$.
The first algorithm for pattern matching in the streaming model was shown in \cite{Porat:2009} and requires $O(\log(m))$ space.
The authors of \cite{Ergun:2010} give a simpler pattern matching algorithm with no preprocessing, as well as a related
streaming algorithm for estimating a stream's Hamming distance to $p$-periodicity.
Breslauer and Galil \cite{Breslauer:11} provide an algorithm which does not report false negatives and can also be run in real-time. 
All of the above algorithms in the string model take advantage of  Karp-Rabin fingerprints \cite{RabinKarp:1987}.
%
\newline
\noindent{\bf Our results} 
In this paper we present three algorithms, \A, \B, and  \aLog  for finding  palindromes and estimating their length in a given stream $S$ of length $n$. 
%
%
We assume that the work space is bounded while the output space is unlimited.\\ 
Given an index $m$ in stream $S$,  $P[m]$ denotes the palindrome of 
maximal length cantered at index $m$ of $S$. 
Our algorithms identify a palindrome $P[m]$ by its {\it midpoint} $m$ and by its length $\ell(m)$.
Our first algorithm  outputs all palindromes in $S$ and therefore solves the \problemDefAll. 
\begin{theorem}[\A]\label{the:firstpassisgood} 	
For any $\varepsilon \in [\sfrac{1}{\sqrt n},1]$  $\algA(S,\varepsilon)$ reports for every palindrome $P[m]$  in $S$ its midpoint $m$ as well as an estimate $\tilde\ell(m)$ (of $\ell(m)$) such that w.h.p.\footnote{We say an event happens \emph{with high probability} (w.h.p.) if its probability is at least $1-{1}/{n^c}$ for $c\in\mathbb{N}$.} $\ell(m) - \varepsilon \sqrt n < \tilde\ell(m) \leq \ell(m).$
The algorithm makes one pass over $S$, uses  $O( \sfrac{n}{\varepsilon})$ time, and $O( \sfrac{\sqrt n}{\varepsilon}  )$ space. 
\end{theorem}
The algorithm can easily be modified to report all palindromes $P[m]$ in $S$ with $\ell(m) \geq t $ and  no $P[m]$ with $\ell(m)<t- \varepsilon\sqrt n$ for some 
threshold $t\in\mathbb{N}$. For $t\leq \sqrt n$ one can modify the algorithm to report a palindrome $P[m]$ if and only if $\ell(m) \geq t $. Note, the algorithm is also $(1+\varepsilon)$-approximative.\\
Our next algorithm, \B, uses two-passes to solve the \problemDefLongest. It uses  \A as the first pass. 
In the second pass the algorithm finds the midpoints of all palindromes of length exactly $\ell_{\max}$ where $\ell_{\max}$ is the (initially unknown) length of the longest palindrome in $S$. 
%
%
\begin{theorem}[\B]\label{the:main}
\algB reports w.h.p. $\ell_{max}$ and $m$ for all palindromes $P[m]$ with a 
length of $\ell_{max}$.
The algorithm makes two passes over $S$,  uses $O(n)$ time, and $O(\sqrt n )$ space.
\end{theorem}
Arguably the most significant contribution of this paper is an algorithm which requires only logarithmic space.
In contrast to \A (Theorem \ref{the:firstpassisgood}) this algorithm has a multiplicative error  and it reports only one of the longest palindromes (see \problemDefLongest) instead of all of them due to the limited space.
\begin{theorem}[\aLog]\label{the:log}
For any $\varepsilon$ in $(0,1]$,
\algLog reports w.h.p. an arbitrary palindrome $P[m]$ of length at least $\ell_{max}/(1+\varepsilon)$. 
The algorithm makes one pass over $S$, uses $O(\frac{n\log(n)}{\varepsilon\log(1+\varepsilon)})$ time, and 
$O(\frac{\log(n)}{\varepsilon\log(1+\varepsilon)})$ 
 space. 
\end{theorem}
We also show two practical generalizations of our algorithms which can be  run simultaneously. These results are presented in the next observation and the next lemma.

\begin{observation}\label{ObsVariant}
For $\ell_{max} \geq \sqrt n$, there is an algorithm which reports w.h.p. the midpoints of all palindromes $P[m]$ with $ \ell(m) > \ell_{max} - \varepsilon \sqrt n$.
The algorithm makes one pass over $S$, uses  $O( \sfrac{n}{\varepsilon})$ time, and $O( \sfrac{\sqrt n}{\varepsilon}  )$ space. 
\end{observation}

\begin{lemma}\label{lem:smallellmax} 
For  $\ell_{max} < \sqrt n$, there is an algorithm which reports w.h.p. $\ell_{max}$ and a $P[m]$ s.t. $\ell(m)=\ell_{max}$.
The algorithm makes one pass over $S$, uses  $O( n)$ time, and $O( \sqrt n )$ space. 
\end{lemma}
Additionally, we show an almost matching bound for the additive error of \algA. In more detail, we will show that
any randomized one-pass algorithm that approximates the length of the longest palindrome up to an additive error of $\varepsilon\sqrt{n}$ with probability $1$ must use $\Omega(\sfrac{\sqrt{n}}{\varepsilon})$ space.
\vspace{-2.5mm}
\section{Model and Definitions}\label{sec:model}\vspace{-1.5mm}
Let $S\in \Sigma^n$ denote the input stream of length $n$ over an alphabet $\Sigma$\footnote{All soundness, space, and time complexity analyses assumes $|\Sigma|$ to be polynomial. One can use a proper random hash function for bigger alphabets.}. For simplicity we assume symbols to be positive integers, i.e., $\Sigma \subset \mathbb{N}$.
We define $S[i]$ as the symbol at index $i$ and $S[i,j]= S[i],S[i+1],\ldots S[j]$.
In this paper we use the streaming model: In one {\it pass} the algorithm goes over the whole input stream $S$, reading $S[i]$ in {\em iteration} $i$ of the pass.
In this paper we assume that the algorithm has a memory of size $o(n)$,
but the output space is unlimited.
We use the so-called word model where the space equals the number of $O(\log(n))$ registers (See \cite{Breslauer:11}).\\
$S$  contains an \emph{odd palindrome} of length $\ell$ with midpoint $m\in 
\{\ell\kdots n-\ell\}$ if $S[m-i]=S[m+i]$ for all $i\in\{1\kdots\ell\}$.
Similarly, $S$ contains an \emph{even palindrome} of length $\ell$ if  
$S[m-i+1]=S[m+i]$ for all $i\in\{1\kdots\ell\}$. In other words, a palindrome is odd if and only if its length is odd.
For simplicity, our algorithms assume  palindromes to be even - it is easy to adjust our results for finding odd palindromes by apply the algorithm to 
$S[1]S[1]S[2]S[2]\cdots S[n]S[n]$ instead of $S[1,n]$.\\
The maximal palindrome (the palindrome of maximal length) in $S[1,i]$ with midpoint $m$ is called $P[m,i]$ and the maximal palindrome in $S$ with midpoint $m$ is called $P[m]$ which equals $P[m,n]$.
We define $\ell(m,i)$ as the maximum length of the palindrome with midpoint $m$ in the substring $S[1,i]$.
The maximal length of the palindrome in $S$ with midpoint $m$ is denoted by $\ell(m)$. 
Moreover, for $z\in  \mathbb{Z}\setminus\{1 \kdots n \}$ we define $\ell(z) = 0$.
Furthermore,  for $\ell^* \in \mathbb{N}$ we define $P[m]$ to be an $\ell^*$-palindrome if $\ell(m)\geq \ell^*$. Throughout this paper, $\tilde\ell()$ refers to an estimate of $\ell()$.\\
We use the KR-Fingerprint, 
which was first defined by Karp and Rabin \cite{RabinKarp:1987} to compress strings and was later used in the streaming pattern matching problem (see \cite{Porat:2009}, \cite{Ergun:2010}, and \cite{Breslauer:11}).
For a string $S'$ we define  the forward fingerprint (similar to \cite{Breslauer:11}) and its reverse as follows.
$\phi_{r,p}^{F}(S') = \left( \sum_{i=1}^{|S'|}{S'[i] \cdot r^i} \right) \mod\ p $ $\phi_{r,p}^{R}(S') = \left( \sum_{i=1}^{|S'|}{S'[i] \cdot r^{l-i+1}} \right) \mod\ p,<$
where $p$ is an arbitrary prime number in $[n^4,n^5]$ 
and $r$ is randomly chosen from $\{1\kdots p\}$. 
We write $\phi^{F}$($\phi^{R}$ respectively) as opposed to $\phi_{r,p}^{F}$($\phi_{r,p}^{R}$ respectively) whenever $r$ and $p$ are fixed.
We define for  $1\leq i\leq j\leq n$ the fingerprint $F^F(i,j)$ as the fingerprint of $S[i,j]$, i.e., 
$F^F(i,j)  =\phi^{F}(S[i,j])= r^{-(i-1)}(\phi^{F}(S[1,j]) - \phi^{F}(S[1,i-1]))\text{ mod }p.$
Similarly,  $F^R(i,j) = \phi^{R}(S[i,j])= \phi^{R}(S[1,j])-r^{j-i+1}\cdot\phi^{R}(S[1,i-1])\text{ mod }p.$
For every $1\le i\le n-\sqrt{n}$ the fingerprints $F^F(1,i-1-\sqrt n)$ and $F^R(1,i-1-\sqrt n)$ are called \mfp.
Note that it is easy to obtain $F^F(i,j+1)$ by adding the term $S[j+1]r^{j+1}$ to $F^F(i,j)$. Similarly, we obtain $F^F(i+1,j)$ by subtracting $S[i]$ from $r^{-1}\cdot F^F(i,j)$. 
The authors of \cite{Breslauer:11} observe useful properties which we state in 
the appendix.
Finally, Section \ref{sec:LB} gives a lower bound  for any algorithm finding the longest palindrome with probability $1$.
\begin{figure}
\centering
\begin{minipage}{.49\textwidth}
  \centering
\includegraphics[scale=0.50]{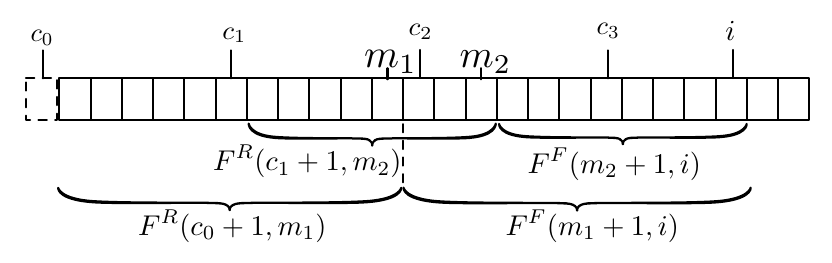}
  \caption{At iteration $i$ two midpoints $m_1$ and $m_2$ are checked.
Corresponding substrings are denoted by brackets.
Note, the distance from  $c_0$ to $m_1$  equals the distance from $m_1$ to $i$.
Similarly, the distance from  $c_1$ to $m_2$ equals the distance from $m_2$ to $i$. }
\label{fig:comparewithcheckpointsinsqrtn}
\end{minipage}
\begin{minipage}{.49\textwidth}
  \centering
\vspace*{+0.27cm}
\includegraphics[scale=0.40, trim = 20mm 0mm 0mm 0mm, clip]{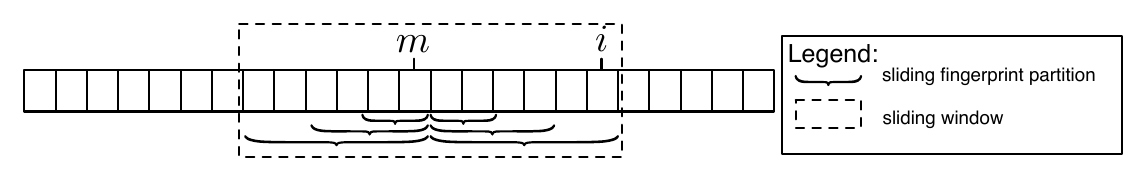}
  \caption{Illustration of the \slidingpairs after iteration $i$ of algorithm with $\sqrt n=6$, $\varepsilon=1/3$, and $m=i-\sqrt n$.   }
 \label{fig:slidingfingerprintpartition}
\end{minipage}
\vspace{-4.5mm}
\end{figure}
\vspace{-4.5mm}
\section{\algsA}\label{section:algsA}\vspace{-1.5mm}
In this section, we introduce a simple one-pass algorithm
which reports all midpoints and length estimates of palindromes in $S$.
Throughout this paper we use $i$ to denote the current index which the algorithm reads. 
\sA keeps the last 
$2\sqrt n$ symbols of $S[1,i]$ in the memory.  \\
It is easy to determine the exact length palindromes of length less than $\sqrt n$ since any
such palindrome is fully contained in memory at some point.
However, in order to achieve a better time bound the algorithm only {\it approximates} the length of short palindromes.  
It is more complicated to estimate the length of a palindrome with a length of at least $\sqrt n$. However,
\sA detects that its length is at least $\sqrt n$ and stores it as an \single (introduced later) in
a list $L_i$. The \single contains the midpoint as well as a length estimate of the palindrome, which is updated as $i$ increases. \\
In order to estimate the lengths of the long palindromes the algorithm designates certain indices of $S$ as \emph{checkpoints}. 
For every checkpoint $c$ the algorithm stores a fingerprint $F^R(1,c)$  enabling the algorithm to do the following. For every midpoint $m$ of a long palindrome:
Whenever the distance from a checkpoint $c$ to $m$ ($c$ occurs before $m$) equals the distance from $m$ to $i$, the algorithm 
compares the substring from $c$ to $m$ to the reverse of the substring from $m$ to $i$ by using fingerprints.
We refer to this operation as \emph{checking} $P[m]$ against checkpoint $c$. 
If $S[c+1,m]^R=S[m+1,i]$, then we say that $P[m]$ was \emph{successfully checked} with $c$ and the algorithm updates the length estimate for $P[m]$, $\tilde\ell(m)$.
The next time the algorithm possibly updates $\tilde\ell(m)$ is after $d$ iterations where $d$ equals the distance between checkpoints.
This distance $d$ gives the additive approximation. See Figure \ref{fig:comparewithcheckpointsinsqrtn} for an illustration.\\
%
%
%
%
%
%
We need the following definitions before we state the algorithm: For $k\in \mathbb{N}$ with $0 \leq k \leq \lfloor \frac{\sqrt n}{\varepsilon} \rfloor$ 
 checkpoint  $c_k$ is the index at position $k\cdot \lfloor\varepsilon \sqrt n\rfloor$ thus
checkpoints are  $\lfloor\varepsilon \sqrt n\rfloor$ indices apart.
Whenever we say that an algorithm stores a checkpoint, this means storing the data belonging to this checkpoint.
Additionally, the algorithm stores \slidingpairs, fingerprints of size
$\lfloor\varepsilon \sqrt n\rfloor, 2\lfloor\varepsilon \sqrt n\rfloor, \ldots$ starting or ending in the middle of the sliding window.
In the following, we first describe the data that the algorithm has in its memory 
after reading $S[1,i-1]$, then we describe the algorithm itself. Let $R_S(m,i)$ denote 
the representation of $P[m]$ which is stored at time $i$.
As opposed to storing $P[m]$ directly, the algorithm stores $m$, $\tilde\ell(m,i)$, $ F^F(1,m)$, and  $F^R(1,m)$.
\newline\noindent{\bf Memory invariants.}
Just before algorithm \sA
reads $S[i]$ it has stored the following information.  Note that, for ease of referencing,
during an iteration $i$ data 
structures are indexed with the iteration number $i$. \\
That is, for instance,   
$L_{i-1}$ is called $L_{i}$ after $S[i]$ is read. \vspace{-1.5mm}
\begin{enumerate}
\item The contents of the sliding window $S[i-2\sqrt n -1,i - 1 ]$. 
\item The two \mfp  $F^F(1,i-1)$ and $F^R(1,i-1)$.
\item A list of \slidingpairs: Let $r$ be the maximum integer s.t. $r\cdot\floor{\varepsilon\sqrt n}<\sqrt n$.\\
For $j\in\{\floor{\varepsilon\sqrt n},2\cdot\floor{\varepsilon\sqrt n}\kdots r\cdot\floor{\varepsilon\sqrt n},\sqrt n\}$
the algorithm stores the pair \\ $F^R((i-\sqrt n) - j ,(i-\sqrt n) -1)$, and $F^F(i-\sqrt n , (i-\sqrt n) + j -1)$. See Figure \ref{fig:slidingfingerprintpartition} for an illustration.

\item A list CL$_{i-1}$ which consists of all fingerprints of prefixes of $S$ ending at
already seen checkpoints, i.e.,  
$\text{CL}_{i-1}=\left[F^R(1,c_1), F^R(1,c_2)\kdots  F^R\big(1,c_{\floor[]{\sfrac{(i-1)}{\floor{\varepsilon \sqrt n}}}}\big)\right]$

\item\label{enu:Listofcandidates} A list $L_{i-1}$ containing representation of all $\sqrt n$-palindromes with a midpoint located in\\ $S[1, (i-1)-\sqrt n]$. The $j^{th}$ entry of $L_{i-1}$ has the form\\ $R_S(m_j,i-1)=(m_j, \tilde\ell(m_j,i-1), F^F(1,m_j) , F^R(1,m_j)) $ where
\begin{enumerate}[(a)]
\item $m_j$ is the midpoint of the $j^{th}$ palindrome in $S[1,(i-1) -\sqrt n]$ with a length of at least $\sqrt n$. Therefore, $m_j<m_{j+1}$ for $1\le j\le  |L_{i-1}|-1$.
\item $\tilde\ell(m_j,i-1)$ is the current estimate
%
of $\ell(m_j,i-1)$.
\end{enumerate}

\end{enumerate}
In the following, we explain how the algorithm maintains the above invariants.
\newline\noindent{\bf Maintenance.}
At iteration $i$ the algorithm performs the following steps. It is implicit that $L_{i-1}$ 
and $CL{i-1}$ become $L_i$ and $CL_i$ respectively. \vspace{-2.7mm}
\begin{enumerate}\itemsep0pt
\item Read $S[i]$, set $m= i-\sqrt n$. Update the sliding window to $S[m-\sqrt n, i ] = 
S[i-2\sqrt n, i]$
\item Update the \mfp to be $F^F(1,i)$ and $F^R(1,i)$.
\item\label{checkpointing} If $i$ is a checkpoint (i.e., a multiple of $\lfloor\varepsilon \sqrt n \rfloor$), then add $F^R(1,i)$ to CL$_i$.
\item\label{alg:onlyoutputlongpalindromes} Update all \slidingpairs:
For $j\in\{\floor{\varepsilon\sqrt n},2\cdot\floor{\varepsilon\sqrt n},\kdots r\cdot\floor{\varepsilon\sqrt n},\sqrt n\}$
\begin{itemize}
\item Update $F^R(m - j ,m -1)$ to $F^R(m - j+1 ,m)$ and $F^F(m , m + j -1)$ to $F^F(m +1, m+j )$.
\item\label{alg:insideslidingwindow} If $F^R(m - j+1 ,m)=F^F(m +1, m+j )$, then set $\tilde\ell(m,i)=j$.
\item\label{step:report short palindromes} If  $\tilde\ell(m,i) < \sqrt n$,
output $m$ and $\tilde\ell(m,i)$.
\end{itemize}
\item \label{enu:NewPalindromeInSlidingWindow}  If  $\tilde\ell(m,i)\geq \sqrt n$, add \emph{add} $R_S(m,i)$ to $L_{i}$:\\
$L_i= L_{i}\ \circ \ (m, \tilde\ell(m,i), F^F(1,m), F^R(1,m))$.
\item\label{enu:CompareWithMidpoints}  For all $c_k$ with $1\leq k \leq \lfloor \frac{i}{\lfloor \varepsilon \sqrt n\rfloor} 
\rfloor$  and $R_S(m_j,i)\in L_i$ with $i-m_j=m_j-c_k$, check if $\tilde\ell(m_j,i)$ can be updated:
\begin{itemize}\itemsep0pt 
\item\label{alg:comparison} If the left side of $m_j$ is the reverse of the right side of $m_j$ (i.e., $F^R(c_k +1, m_j)=F^F(m_j +1, i)$) then update $R_S(m_j,i)$ by updating $\tilde\ell(m_j,i)$ to $i-m_j$. 
\end{itemize}
\item\label{sALaststep} If $i=n$, then report $L_n$.
\end{enumerate} \vspace{-1.5mm}
In all proofs in this paper which hold w.h.p. we assume that fingerprints do not fail as we take less than $n^2$ fingerprints and by using the following Lemma, the probability that a fingerprint fails is at most $1/n^{c+2}$. \vspace{-1.4mm}
\begin{lemma}\label{lem:breslauer2}
(Theorem 1 of \cite{Breslauer:11}) For two arbitrary strings $s$ and $s'$ with $s \not= s'$ the probability that $\phi^F(s) = \phi^F(s')$ is smaller than $\sfrac{1}{n^{c+2}}$ for some $c\in\mathbb{N}$.
\end{lemma} 
Thus, by applying the union bound the probability that no fingerprint fails is at least $1-n^{-c}$.
The following lemma shows that the \sA finds all palindromes along with the estimates as stated in Theorem \ref{the:firstpassisgood}.
\sA does not fulfill the time and space bounds of Theorem \ref{the:firstpassisgood}; we will later show how to improve its 
efficiency. The proof can be found in the appendix.
\begin{lemma}\label{lem:invariantsgood}
For any $\varepsilon$ in  $[\sfrac{1}{\sqrt n},1]$  $\A(S,\varepsilon)$ reports for every palindrome $P[m]$ in $S$ its midpoint $m$ as well as an estimate $\tilde\ell(m)$  such that w.h.p. $\ell(m) - \varepsilon \sqrt n < \tilde\ell(m) \leq \ell(m).$
\end{lemma} \vspace{-3.5mm}
\section{A space-efficient version}\label{section:algA}\vspace{-1.5mm}
In this section, we show how to modify \sA so that it matches the time and space requirements of Theorem \ref{the:firstpassisgood}. 
The main idea of the space improvement is to store the lists $L_i$ in a compressed form.
\newline
\noindent{\bf Compression}\label{def:runtuple}
It is possible in the simple algorithm for $L_i$ to have linear length. 
In such cases $S$ contains many overlapping palindromes which show a certain {\it periodic}
pattern as shown in Corollary \ref{cor:runhasstructure}, which our algorithm exploits
to compress the entries of $L_i$. This idea was first introduced in \cite{Porat:2009}, and is used in \cite{Ergun:2010}, and \cite{Breslauer:11}. More specifically, our technique is a modification of the compression in \cite{Breslauer:11}.
In the following, we give some definitions in order to show how to compress the list.
First we define a \emph{run} which is a sequence of midpoints of overlapping palindromes.
\begin{definition}[$\ell^*-$Run]\label{def:compressible}
Let $\ell^*$ be an arbitrary integer and $h\geq 3$. Let $m_1,m_2,m_3\kdots m_h$ be consecutive midpoints of $\ell^*$-palindromes in $S$. $m_1 \kdots m_h$ form an \emph{$\ell^*$-run} if $m_{j+1}-m_j\leq \ell^*/2$ for all $j\in\{1\kdots h-1\}$.
\end{definition}
In Corollary \ref{cor:runhasstructure} we show that $m_2-m_1= m_3-m_2 = \cdots = m_h - m_{h-1}$. 
We say that a run is maximal if the run cannot be extended by other palindromes. More formally:
\begin{definition}[Maximal $\ell^*-$Run]\label{defMaximalRun} 
An $\ell^*$-run over $m_1\kdots m_h$  is \emph{maximal} it satisfies both of the following:
\begin{inparaenum}[i)]
\item  $\ell(m_1 - (m_2 - m_1)) < \ell^*,$
\item $\ell(m_h + (m_2 - m_1)) < \ell^*.$
\end{inparaenum}
\end{definition}
\sA stores palindromes explicitly in $L_i$, i.e., $L_{i}=\ [R_S(m_1,i); \dots ; R_S(m_{|L_{i}|},i) ]$ where
$R_S(m_j,i)=(m_j, \tilde\ell(m_j,i), F^F(1,m_j) , F^R(1,m_j))$, for all $j\in \{1,2\kdots h\}$.
The improved 
\algA stores these midpoints in a compressed way in list $\hat L_i$.
 \A
distinguishes among three cases: Those palindromes which \vspace{-2.7mm}
\begin{enumerate}
\item are not part of a $\sqrt n$-run are stored explicitly as before. We call them \singles. Let $P[m,i]$ be such a palindrome. After iteration $i$ the algorithm stores $R_S(m,i)$.
\item form a maximal $\sqrt n$-run are stored in a data structure called \finished. Let  $m_1\kdots m_h$ be the midpoints of a maximal $\sqrt n$-run.
The data structure stores the following information. 
\begin{itemize}
\item $m_1$, $m_{2} - m_{1}$, $h$, $\tilde\ell(m_1, i)$, $\tilde\ell(m_{\left\lfloor \frac{ 1+h}{2} \right\rfloor},i)$, $\tilde\ell(m_{\left\lceil \frac{ 1+h}{2} \right\rceil}, i)$, $\tilde\ell(m_h, i)$,
 \item $F^F(1,m_1)$, $F^R(1,m_1)$, $F^F(m_1 + 1,m_{2} )$, $F^R(m_1+1,m_{2})$
\end{itemize}
\item form a $\sqrt n$-run which is not maximal (i.e., it can possibly be extended) in a data structure called \current.
The information stored in an \current is the same as in an \finished, but it does not contain the entries:  
$\tilde\ell(m_{\left\lfloor \frac{ 1+h}{2} \right\rfloor},i)$, $\tilde\ell(m_{\left\lceil \frac{ 1+h}{2} \right\rceil}, i)$, and $\tilde\ell(m_h, i)$.
\end{enumerate}
The algorithm stores only the estimate (of the length) and the midpoint of the following palindromes explicitly. 
\begin{itemize}
\item $P[m]$ for an \single (Therefore all palindromes which are not part of a $\sqrt n$-run)
\item $P[m_1]$, $P[m_{\lfloor \sfrac{(h+1)}{2}\rfloor}]$, $P[m_{\lceil  \sfrac{(h+1)}{2} \rceil}]$, and $P[m_h]$ for an \finished
\item $P[m_1]$ for an \current. 
\end{itemize}
In what follows we refer to the above listed palindromes as \emph{explicitly stored} palindromes.
We argue in Observation \ref{obs:constantNumberOfCandi} that in any interval of length $\sqrt n$ the number of \emph{explicitly stored} palindromes is bounded by a constant.
%
%
%
\vspace{-2.7mm}
\subsection{\algA}\label{sub:sublinearversion}  \vspace{-1.5mm}
In this subsection, we describe some modifications of \sA in order to obtain a space complexity of $O(\frac{\sqrt n}{\varepsilon})$ and a total running time of  $O(\frac{n}{\varepsilon})$. 
\A is the same as \sA, but it compresses the stored palindromes.
\A uses the same memory invariants as \sA, but it uses $\hat L_i$ as opposed to $L_i$.\\
\A uses the first four steps of \sA. Step \ref{enu:NewPalindromeInSlidingWindow}, Step \ref{enu:CompareWithMidpoints}, and Step \ref{sALaststep} are replaced. 
The modified Step \ref{improvement1} ensures that there are at most two \singles per interval of length $\sqrt n$.
Moreover, Step \ref{enu:CompareWithMidpoints} is adjusted since \A stores only the length estimate of explicitly stored palindromes.

\begin{enumerate}
\setcounter{enumi}{4}
\item\label{improvement1} 
If  $\tilde\ell(m,i)\geq \sqrt n$, obtain $\hat L_{i}$ by adding the palindrome with midpoint $m(=i-\sqrt n)$ to $\hat L_{i-1}$ as follows:
\begin{enumerate}
\item The last element in $\hat L_i$ is the following 
\current \newline $\big(m_1,m_{2} - m_{1} ,h, \tilde\ell(m_1, i), F^F(1,m_1),F^R(1,m_1 ),F^F(m_1 + 1,m_{2} ),F^R(m_1+1,m_{2})\big).$ 
\begin{enumerate} 
\item\label{stepaddtorun} If the palindrome can be added to this run, i.e., $m=m_1+h(m_2-m_1)$,   
then we increment the $h$ in the \current by 1.
\item\label{enu:AlgoFinishesRun} 
If the palindrome cannot be added: 
Store $P[m,i]$ as an \single: \\
$\hat L_i= \hat L_{i}\ \circ \ (m, \tilde\ell(m,i), F^F(1,m), F^R(1,m))$.
Moreover, convert the \current into the \finished by adding
$\tilde\ell(m_{\left\lfloor \frac{ 1+h}{2} \right\rfloor},i)$, $\tilde\ell(m_{\left\lceil \frac{ 1+h}{2} \right\rceil}, i)$ and $\tilde\ell(m_h, i)$:
First we calculate $m_{\lfloor \frac{ 1+h}{2} \rfloor}= m_1+\left(\lfloor \frac{ 1+h}{2} - 1 \rfloor \right)(m_2-m_1).$ One can calculate
$m_{\lceil \frac{ 1+h}{2} \rceil}$
similarly.
For $m'\in \{m_{\left\lfloor \frac{ 1+h}{2} \right\rfloor}, m_{\left\lceil \frac{ 1+h}{2} \right\rceil}, m_h \}$  calculate
$\tilde\ell(m',i)=$\\$\underset{i-2\sqrt n \leq j \leq i}{\max}\{ j-m'\ \mid\ \exists c_{k}\text{ with }j-m' = m'-c_k \text{ and } F^R(c_k + 1,m')=F^F(m' + 1,j)  \}.$
\end{enumerate}

\item The last two entries in $\hat L_i$ are stored as \singles and together with $P[m,i]$ form a $\sqrt n$-run. Then remove the entries of the two palindromes out of $\hat L_{i-1}$ and add a new \current with all three palindromes to $\hat L_{i-1}$:\\
 $m_1,\tilde\ell(m_1, i), F^F(1,m_1),F^F(1,m_2),F^R(1,m_1 ),F^R(1,m_2 ),m_{2} - m_{1} , h=3 $. 
Retrieve $F^F(m_1 + 1,m_{2} )$ and $F^R(m_1+1,m_{2})$.


\item Otherwise, store $P[m,i]$ as an \single: 
$\hat L_i= \hat L_{i}\ \circ \ (m, \tilde\ell(m,i), F^F(1,m), F^R(1,m))$
\end{enumerate}

\item This step is similar to step \ref{enu:CompareWithMidpoints} of \sA the only difference is that we check only for explicitly stored palindromes if they can be extended outwards. \footnote{This step is only important for the running time.}
\item If $i=n$. If the last element in $\hat L_i$ is an \current, then convert it into an \finished as in \ref{enu:AlgoFinishesRun}. Report $L_n$.
\end{enumerate}  
\vspace{-2.5mm}
\subsection{Structural Properties}\vspace{-1.5mm}
In this subsection, we prove structural properties of palindromes.
These properties allow us to compress (by using \singles and \finisheds) overlapping palindromes $P[m_1] \kdots P[m_h]$ in such a way that at any iteration $i$ all the information stored $R_S(m_1,i) \kdots R_S(m_h,i)$ is available. The structural properties imply, informally speaking, 
that the palindromes are either far from each other, leading to a small number of them, or they are overlapping and it is possible to compress them.
Lemma \ref{lem:midpointshavestructure} shows this structure for short intervals containing at least three palindromes. Corollary \ref{cor:runhasstructure} shows a similar structure for palindromes of a run which is used by \A.
We first give the common definition of periodicity. 
\begin{definition} [\emph{period}]
A string $S'$ is said to have period $p$ if it consists of repetitions
of a block of $p$ symbols.
Formally, $S'$ has period $p$ if $S'[j] = S[j+p] \text{ for
all }j = 1\kdots |S'| - p$. \footnote{ Here, $p$ is called a period for $S'$ even if $p>|S'|/2$}
\end{definition}
\begin{lemma}\label{lem:midpointshavestructure}
Let $m_1 < m_2 < m_3 < \dots < m_h$ 
be indices in $S$ that are consecutive midpoints of $\ell^*$-palindromes for an arbitrary natural number $\ell^*$. 
If $m_{h}-m_1\leq \ell^*$, then
\begin{enumerate}[(a)]
\item $m_1,m_2,m_3\kdots m_h$ are equally spaced in $S$, i.e., $|m_{2}-m_{1}|=|m_{k+1}-m_{k}|$  $\forall k\in\{1\kdots h-1\}$
\item  $S[m_1+1,m_h]=\begin{cases} (ww^R)^{\frac{h-1}{2}} & h\text{ is odd }\\ (ww^R)^{\frac{h-2}{2}}w & h \text{ is even } \end{cases}$, where $w=S[m_1+1,m_2]$.
\end{enumerate}
\end{lemma}
\begin{proof}
Given $m_1,m_2 \kdots  m_h$ and $\ell^*$ we prove the following stronger claim by induction over the midpoints $\{m_1\kdots m_j\}$.
\begin{inparaenum}[(a')]
\item $m_1,m_2\kdots m_j$ are equally spaced.\label{itm:sta}
\item  $S[m_1+1,m_j+\ell^*]$ is a prefix of $ww^Rww^R...$ .\label{itm:stb}
\end{inparaenum}\\
\emph{Base case} $j=2$:
Since we assume $m_1$ is the midpoint of an $\ell^*$-palindrome and $\ell^*\geq m_h-m_1\geq m_2-m_1=|w|$,
we have that $S[m_1-|w|+1,m_1]=w^R$.
Recall that $\ell(m_2) \geq \ell^*\geq|w|$ and thus, $S[m_1+1,m_2+|w|]=ww^R$.\\ 
We can continue this argument and derive that 
$S[m_1 + 1, m_2 +\ell^*]$ is a prefix of $ww^R\dots ww^R$.
(\ref{itm:sta}') for $j=2$ holds trivially.\\
\emph{Inductive step} $j-1\rightarrow j $:
Assume (\ref{itm:sta}') and (\ref{itm:stb}') hold up to $m_{j-1}$. We first argue that $|m_j-m_1|$ is  a multiple of $|m_2-m_1|=|w|$.
Suppose $m_j=m_1+|w|\cdot q+r$ for some integers $q\geq 0$ and $r\in \{1\kdots|w|-1\}$.
Since $m_j\leq m_{j-1}+\ell^*$, the interval $[m_1+1,m_{j-1}+\ell^*]$ contains $m_j$.
Therefore, by inductive hypothesis, $m_j-r$ is an index where either $w$ or $w^R$ starts.
This implies that the prefix of $ww^R$(or $w^Rw$) of size $2r$ is  a palindrome and the string $ww^R$(or $w^Rw$) has period  $2r$.
On the other hand, by consecutiveness assumption, there is no midpoint of an $\ell^*$-palindrome in the interval $[m_1+1,m_2-1]$.
does not have a period of $2p$, a contradiction. We derive that $m_j-m_1$ is multiple of $|w|$. \\
Hence, we assume $m_{j}=m_{j-1}+q\cdot |w|$ for some $q$.
The assumption that $m_j$ is a midpoint of an $\ell^*$-palindrome beside the inductive hypothesis implies (\ref{itm:stb}') for $j$.
The structure of $S[m_{j-1}+|w|-\ell^*+1,m_{j-1}+|w|+\ell^*]$ shows that $m_{j-1}+|w|$ is a midpoint of an $\ell^*$-palindrome.
This means that $m_j=m_{j-1}+|w|$.
This gives (\ref{itm:sta}') and yields the induction step.
\end{proof} \vspace{-1.33mm}
\vspace{-1.8mm}
Corollary \ref{cor:runhasstructure} shows the structure of overlapping palindromes and is essential for the compression.
The main difference between Corollary \ref{cor:runhasstructure} and Lemma \ref{lem:midpointshavestructure} is the required distance between the midpoints of a run. Lemma \ref{lem:midpointshavestructure} assumes that every palindrome in the run overlaps with all other palindromes. 
In contrast, Corollary \ref{cor:runhasstructure} assumes that every palindrome $P[m_j]$ overlaps with $P[m_{j-2}]$, $P[m_{j-1}]$, $P[m_{j+1}]$, and $P[m_{j+2}]$. It can be proven by an induction over the midpoints and using Lemma \ref{lem:midpointshavestructure}. The proof is in the appendix.\vspace{-1.5mm}
\begin{corollary}\label{cor:runhasstructure}
If  $m_1,m_2\kdots m_h$ form an $\ell^*$-run for an arbitrary natural number $\ell^*$ then \vspace{-0.5mm}
\begin{enumerate}[(a)]
\item $m_1,m_2,m_3\kdots m_h$ are equally spaced in $S$, i.e., $|m_{2}-m_{1}|=|m_{k+1}-m_{k}|$  $\forall k\in\{1\kdots h-1\}$
\item  $S[m_1+1,m_h]=\begin{cases} (ww^R)^{\frac{h-1}{2}} & h\text{ is odd }\\ (ww^R)^{\frac{h-2}{2}}w & h \text{ is even } \end{cases}$, where $w=S[m_1+1,m_2]$.
\end{enumerate}
\end{corollary}
Lemma \ref{lem:palindromesOnOneSideOfRun} shows the pattern for the lengths of the palindromes in each half of the run. This allows us to only store a constant number of length estimates per run.
The proof can be found in the appendix.
\begin{lemma}\label{lem:palindromesOnOneSideOfRun}
At iteration i, let $m_1,m_2,m_3,...,m_h$ be midpoints of a maximal $\ell^*$-run in $S[1,i]$ for an arbitrary natural number $\ell^*$. For any midpoint $m_j$, we have:\\
$\ell(m_j,i)= 
	\begin{cases} 
	   \ell(m_1,i)+(j-1)\cdot (m_2-m_1)   &  j< {\frac{h+1}{2}} \\
	   \ell(m_h,i)+(h-j)\cdot (m_2-m_{1}) &  j> {\frac{h+1}{2}} 
	\end{cases}$
\end{lemma}

\vspace{-2.5mm}
\subsection{Analysis of the Algorithm}\label{sub:analysisA}\vspace{-1.5mm}
We show that one can convert \singles into a run and vice versa and \A's maintenance of \finisheds and \currents does not impair the length estimates.
%
The following lemma shows that one can retrieve the length estimate of a palindrome as well as its fingerprint from an \finished. \vspace{-1.5mm}
\begin{lemma}\label{lem:runconstantsize}
At iteration $i$, the \finished over $m_1, m_2\kdots m_h$ is a lossless compression of \\
$\ [R_S(m_1,i); \dots ; R_S(m_{h},i) ]$
\end{lemma}
Let \comp be the general term for \finished and \current.
We argue that in any interval of length $\sqrt n$ we only need to store at most two single palindromes and two \comps.
Suppose there were three \singles, then, by Corollary \ref{cor:runhasstructure}, they form a $\sqrt n$-run since they overlap each other. Therefore, the three \singles would be stored in a \comp. For a similar reason there cannot be more than two \comps in one interval of length $\sqrt n$. We derive the following observation.
\begin{observation}\label{obs:constantNumberOfCandi}
For any interval of length $\sqrt n$ there can be at most two \singles and two \comps in $L^*$. 
\end{observation}
We now have what we need in order to prove Theorem \ref{the:firstpassisgood}; the proof is given in the appendix.\vspace{-1.5mm}
%
%
%
%
%
\vspace{-3.5mm}
\section{\algB}\label{section:algB}\vspace{-1.5mm}
This section describes  \algB which determines the exact length of the longest palindrome in $S$ using $O(\sqrt n)$ 
space and two passes over $S$.\\
For the first pass this algorithm runs \A$(S,\frac{1}{2})$ (meaning that $\varepsilon=\sfrac{1}{2}$) and
the variant of \A described in Lemma \ref{lem:smallellmax} simultaneously.
The first pass returns $\ell_{max}$ (Lemma  \ref{lem:smallellmax}) if $\ell_{max}<\sqrt n$.
Otherwise, the first pass (Theorem \ref{the:firstpassisgood})  returns for every palindrome $P[m]$,
with $\ell(m) \geq \sqrt n$, an estimate satisfying $\ell(m) - \sfrac{\sqrt n}{2}  < \tilde\ell(m) \leq \ell(m)$ w.h.p..
\\
The algorithm for the second pass is determined by the outcome of the first pass.
For the case $\ell_{max}<\sqrt n$, it uses the sliding window to find all $P[m]$ with $\ell(m)=\ell_{max}$.
If  $\ell_{max}\geq\sqrt n$, then the first pass only returns an additive $\sfrac{\sqrt n}{2}$-approximation of the palindrome lengths.
We define the \emph{uncertain intervals} of $P[m]$ to be: $
I_1(m)=S[m-\tilde\ell(m)-\sfrac{\sqrt n}{2}+1,m-\tilde\ell(m)]$ and
$I_2(m)=S[m+\tilde\ell(m)+1,m+\tilde\ell(m)+\sfrac{\sqrt n}{2}]. 
$
The algorithm uses the length estimate calculated in the first pass to delete all \singles (Step \ref{alg:deleteshortones}) which cannot be the longest palindromes. Similarly, the algorithm (Step \ref{alg:onlykeepthemiddle}) only keeps the middle entries of \finisheds since these are the longest palindromes of their run (A proof can be found in the appendix).
In the second pass, \algB stores $I_1(m)$ for a palindrome $P[m]$ if it was not deleted. 
\algB compares the symbols of $I_1(m)$ symbol by symbol to $I_2(m)$ until the first mismatch is found. Then the algorithm knows the exact length $\ell(m)$  and discards $I_1(m)$.
The analysis will show, at any time the number of stored uncertain intervals is bounded by a constant.
\newline
\noindent{\bf First Pass}
Run the following two algorithms simultaneously:\vspace{-1.5mm}
\begin{enumerate}
\item \A$(S,\sfrac{1}{2})$. Let $L$ be the returned list. 
\item Variant of \A (See Lemma \ref{lem:smallellmax}) which reports $\ell_{max}$ if $\ell_{max}< \sqrt n$. 
\end{enumerate}
\noindent{\bf Second Pass}\vspace{-1.5mm}
\begin{itemize}
\item{$\ell_{max}< \sqrt n$:}
Use a sliding window of size $2\sqrt n$ and maintain two fingerprints $F^R[i-\sqrt n-\ell_{max}+1, i-\sqrt n]$, and $F^F[i-\sqrt n+1,i-\sqrt n+\ell_{max}]$. Whenever these fingerprints match, report $P[i-\sqrt n]$.
\item{$\ell_{max}\geq \sqrt n$:} In this case, the algorithm uses a preprocessing phase first. 
\\
{\bf Preprocessing }\vspace{-1.5mm}
\begin{enumerate}
\item Set $\tilde\ell_{max}=\max\{\tilde\ell(m)\mid P[m] \text { is stored in $L$ as an $R_F$ or an $R_S$ entry}\}$.
\item\label{alg:onlykeepthemiddle} For every \finished $R_F$ in $L$ with midpoints  
$m_1 \kdots m_h$ 
remove $R_F$ from $L$ and add\\
$R_s(m,i)=(m, \tilde\ell(m), F^F(1,m),F^R(1,m))$ to $L$, for $m\in \{ m_{\floor{(h+1)/2}}, m_{\ceil{(h+1)/2}} \}$.
To do this, calculate $m_{\lfloor \frac{ 1+h}{2} \rfloor}= m_1+(\lfloor \frac{ 1+h}{2} \rfloor -1)(m_2-m_1)$ and
 $m_{\lceil \frac{ 1+h}{2} \rceil}= m_1+(\lceil\frac{ 1+h}{2} \rceil -1)(m_2-m_1)$.
Retrieve $F^F(1,m)$ and $F^R(1,m)$  for $m\in \{ m_{\floor{(h+1)/2}}, m_{\ceil{(h+1)/2}} \}$.
\item\label{alg:deleteshortones} Delete all \singles $(m_{k}, \tilde\ell(m_{k}), F^F(1,m_k),F^R(1,m_k))$ with $ \tilde\ell(m_k) \leq \tilde\ell_{max} - \sfrac{ \sqrt n}{2}$ from $L$.
\item For every palindrome $P[m]\in L$ set $I_1(m) :=  (m-\tilde\ell(m)-\sfrac{1}{2}\sqrt n,m-\tilde\ell(m)]$ and set $finished(m)$ := false.
\end{enumerate}
The resulting list is called $L^*$. \\
{\bf String processing}
At iteration $i$ the algorithm performs the following steps.
\begin{enumerate}
\item Read $S[i]$.  If there is a palindrome $P[m]$ such that $i\in$ $I_1(m)$, then store $S[i]$. 
\item\label{alg:lengthdetermined} If there is a midpoint $m$ such that $m+\tilde\ell(m) < i <m+\tilde\ell(m)+\frac{\sqrt n}{2}$, finished($m$) = false, and $S[m-(i-m)+1]\not=S[i]$, then set finished$(m)$ := true and $\ell(m)= i-m-1$.
\item 
If there is a palindrome $P[m]$ such that $i \geq \tilde\ell(m) + m + \frac{\sqrt n}{2}$, then discard $I_1(m)$.
\item\label{alg:laststepofexact} If $i=n$, then output  $\ell(m)$ and $m$ of all $P[m]$ in $L^*$ with $\ell(m)=\ell_{max}$.
\end{enumerate}
\end{itemize}
We analyze \B in the appendix.\vspace{-2.5mm}
\section{\algLog}\label{section:alglog}\vspace{-2.5mm}
In this section, we present an algorithm which reports one of the longest palindromes and uses only logarithmic space.
\aLog  has a multiplicative error instead of an additive error term.
Similar to \A we have special indices of $S$ designated as checkpoints
that we keep along with some constant size data in memory. The checkpoints are used to estimate the length of palindromes.
However, this time checkpoints (and their data) are only stored for a limited time.
Since we move from additive to multiplicative error we do not need checkpoints to be spread evenly in $S$.
At iteration $i$, the number of checkpoints in any interval of fixed length decreases exponentially with distance to $i$. 
%
The algorithm stores a palindrome $P[m]$ (as an \single or \current) until there is a checkpoint $c$ such that $P[m]$ was checked unsuccessfully against $c$. 
A palindrome is stored in the lists belonging to the last checkpoint with which is was checked successfully. 
In what follows we set $\delta\triangleq\sqrt{1+\varepsilon}-1$ for the ease of notation.
Every checkpoint $c$ has an attribute called $level(c)$.
It is used to determine the number of iterations the checkpoint data remains in the memory.
\newline\noindent{\bf Memory invariants.}
After algorithm \aLog
has processed $S[1,i-1]$ and before reading $S[i]$
it contains the following information:
\begin{enumerate}
\item Two \mfp up to index $i-1$, i.e., $F^F(1,i-1)$ and $F^R(1,i-1)$.
\item A list of checkpoints CL$_{i-1}$. 
For every $c\in $ CL$_{i-1}$ we have 
\begin{itemize}
\item $level(c)$ such that $c$ is in CL$_{i-1}$ iff $c\geq (i-1)-2(1+\delta)^{level(c)}$.
\item $fingerprint(c)=F^R(1,c)$
\item a list $L_c$. It contains all palindromes which were successfully checked with $c$, but with no other checkpoint $c' < c$. The palindromes in $L_c$ are either \singles or \currents (See \algA). 
\end{itemize} 
\item The midpoint $m_{i-1}^*$ and the length estimate  $\tilde\ell(m_{i-1}^*,i-1)$ of the longest palindrome found so far. 
\end{enumerate}
%
The algorithm maintains the following property.
If $P[m,i]$ was successfully checked with checkpoint $c$ but with no other checkpoint $c' < c$, then the palindrome is stored in $L_c$. 
The elements in $L_c$ are ordered in increasing order of their midpoint.
The algorithm stores palindromes as \singles and \currents. This time however, the length estimates are not maintained. Adding a palindrome to a current run works exactly (the length estimate is not calculated) as described in \algA.
\newline\noindent{\bf Maintenance. } 
At iteration $i$ the algorithm performs the following steps.\vspace{-2mm}
\begin{enumerate}
	\item Read S[i]. Update the \mfp to be $F^F(1,i)$ and $F^R(1,i)$.
	\item \label{log:maintaincheckpoints} For all $k\geq k_0=\sfrac{\log(1/\delta)}{\log(1+\delta)}$(The algorithm does not maintain intervals of size 0.)
	\begin{enumerate}
		\item If $i$ is a multiple of $\lfloor\delta(1+\delta)^{k-2} \rfloor$, then add the checkpoint $c=i$ (along with the checkpoint data) to CL$_i$.
			Set $level(c)=k$, $fingerprint(c)=F^R(1,i)$ and $L_c=\emptyset$. 
		\item\label{item:removeoldcheckpoints} If there exists a checkpoint $c$ with
			 $level(c)=k$ and $c < i-2(1+\delta)^k$, 
			then  prepend $L_c$ to $L_{c'}$ where $c'=max\{c''\ |\ c''\in$ CL$_i \ and\ c''>c\}$. Merge and create runs in $L_c$ if necessary (Similar to step \ref{improvement1} of \A). Delete $c$ and its data from CL$_i$. 
 	\end{enumerate}

	\item\label{item:logncomparingstep} For every checkpoint $c\in$ CL$_i$
	\begin{enumerate}
		\item  Let $m_c$ be the midpoint of the first entry in $L_c$ and $c'=max\{c''\ |\ c''\in$ CL$_i \ and\ c''<c\}$.
		
			If $i-m_c=m_c-c'$, then we \emph{check} $P[m]$ against $c'$ by doing the following: 
		\begin{enumerate}\itemsep0pt 
			\item If the left side of $m_c$ is the reverse of the right side of $m_c$ (i.e., $F^R(c', m_c)=F^F(m_c, i)$)
 				then move $P[m_c]$ from $L_c$ to $L_{c'}$ by adding $P[m_c]$ to $L_{c'}$:
				\begin{enumerate}
				\item If  $|L_{c'}| \leq 1$, store $P[m_c]$ as a \single.
				\item If $|L_{c'}| = 2$, create a run out of the \singles stored in $L_{c'}$ and $P[m_c]$.
				\item Otherwise, add $P[m_c]$ to the \current in  $L_{c'}$. 
				\end{enumerate}
			\item If the left side of $m_c$ is \emph{not} the reverse of the right side of $m_c$, then remove $m_c$ from $L_c$.	
			\item  If $i-m_c > \tilde\ell(m_i^*)$, then set $m_i^* = m_c$  and set $\tilde\ell(m_i^*) = i-m_c$.\label{item:mStarUpdate}
		\end{enumerate}
	\end{enumerate}
	\item\label{logreport} If $i=n$, then report $m_{i}^*$ and $\tilde\ell(m_{i}^*)$.
\end{enumerate}
\vspace{-3.5mm}
\subsection{Analysis}\vspace{-1.5mm}
\aLog relies heavily on the interaction of the following two ideas.
The pattern of the checkpointing and the compression which is possible due to the properties of overlapping palindromes (Lemma \ref{lem:midpointshavestructure}).
On the one hand the checkpoints are close enough so that the  length estimates are accurate (Lemma \ref{lem:LogAccurateEstimation}). The closeness of the checkpoints ensures that palindromes which are stored at a checkpoint form a run (Lemma \ref{lem:LcFormsarun}) and therefore can be stored in constant space.
On the other hand the checkpoints are far enough apart so that the number of checkpoints and therefore the required space is logarithmic in $n$. \\
%
We start off with an observation to characterize the checkpointing. 
Step \ref{log:maintaincheckpoints} of the algorithm creates a checkpoint pattern: 
Recall that the level of a checkpoint is determined when the checkpoint and its data are added to the memory.
The checkpoints of every level have the same distance.
A checkpoint (along with its data) is removed if its distance to $i$ exceeds a threshold which depends on the level of the checkpoint.
Note that one index of $S$ can belong to different levels and might therefore be stored several times. The following observation follows from Step \ref{log:maintaincheckpoints} of the algorithm.
\begin{observation}\label{obs:spacingcheckpoints} 
At iteration $i$, $\forall k \geq k_0 =\ceil*{ \frac{log(\frac{(1+\delta)^2}{\delta})}{log(1+\delta)}}$.
Let $C_{i,k}=\{c\in$ CL$_i\ |\ level(c)=k \}$.
\begin{enumerate}  \vspace{-4.4mm}
\item $C_{i,k} \subseteq [i-2(1+\delta)^k, i]$.
\item The distance between two consecutive checkpoints of $C_{i,k}$ is $\lfloor  \delta(1+\delta)^{k-2}\rfloor$.
\item $|C_{i,k}|=\ceil[\Big]{\frac{2(1+\delta)^k}{\lfloor  \delta(1+\delta)^{k-2 }\rfloor}} $.
\end{enumerate}
\end{observation}
This observation can be used to calculate the size of the checkpoint data which the algorithm stores at any time. \begin{lemma}\label{lem:numberOfCheckpoints}
At Iteration $i$ of the algorithm the number of checkpoints is in $O\left(\frac{\log(n)}{\varepsilon\log(1+\varepsilon)}\right)$.
\end{lemma}
\begin{proof}
The distance between consecutive checkpoints of level $k$ is $\lfloor  \delta(1+\delta)^{k-2 }\rfloor$ and thus the number of checkpoints per stage is bounded by $$\left\lceil\frac{2(1+\delta)^k}{\lfloor  \delta(1+\delta)^{k-2 }\rfloor}\right\rceil 
\leq \left\lceil\frac{4(1+\delta)^k}{  \delta(1+\delta)^{k-2 }}\right\rceil 
\leq \frac{4(1+\delta)^2}{  \delta} + 1
= \frac{4(1+\varepsilon)}{  \sqrt{1+\varepsilon}-1} + 1 $$ $$
=\frac{4(1+\varepsilon)( \sqrt{1+\varepsilon}+1)}{ \varepsilon } + 1 
\leq \frac{24}{ \varepsilon} + 1$$ 
where the first inequality comes from the fact that $k\geq k_0$.
The number of levels is $\lceil\log_{1+\delta}(n)\rceil = \lceil 2\log_{1+\varepsilon}(n)\rceil$ and the number of checkpoints is therefore bounded by $(2\log_{1+\varepsilon}(n) +1) (\sfrac{24}{\varepsilon} + 1)=O\left(\frac{\log(n)}{\varepsilon\log(1+\varepsilon)}\right)$. The required space to store a checkpoint along with its data is constant.
\end{proof}
The space bounds of Theorem \ref{the:log} hold due to the following property of the checkpointing:
If there are more than three palindromes stored in a list $L_c$ for checkpoint $c$, then the palindromes form a run and can be stored in constant space as the following lemma shows. 
\begin{lemma}\label{lem:LcFormsarun}
At iteration $i$, let $c\in$ CL$_i$ be an arbitrary checkpoint. The list $L_c$ can be stored in constant space.
\end{lemma} 
\begin{proof}
We fix an arbitrary $c\in$ CL$_i$. For the case that there are less than three palindromes belonging to $L_c$, they can be stored 
as \singles in constant space.
Therefore, we assume the case where there are at least three palindromes belonging to $L_c$ and we show that they form a run.
Let $c'$ be the highest (index) checkpoint less than $c$, i.e., $c'=max\{c''\ |\ c''\in$ CL$_i \ and\ c''<c\}$.
We disregard the case that the index of $c$ is 1.
Let $k$ be the minimum value such that $(1+\delta)^{k-1}< i-c \leq (1+\delta)^{k}$.
Recall that $L_c$ is the list of palindromes which the algorithm has successfully checked against $c$ and not against $c'$ yet.
Let $P[m]$ be a palindrome in $L_c$. Since it was successfully checked against $c$ we know that $i-m \geq m-c$. Similarly, since $P[m]$ was not checked against $c'$ we have $i-m < m-c'$. 
Thus, for every $P[m]$ in $L_c$ we have $\frac{i+c'}{2} < m\leq\frac{i+c}{2}$.
Therefore, all palindromes stored in $L_c$ are in an interval of length less than $\frac{i+c}{2} -\frac{i+c'}{2}=\frac{c-c'}{2}$.
If we show that $\ell(m)\geq \frac{c-c'}{2}$ for all $P[m]$ in $L_c$,
then applying Lemma \ref{lem:midpointshavestructure} with $\ell^*=\frac{c-c'}{2}$ on the palindromes in $L_c$ implies that they are forming a run. The run can be stored in constant space in an \current.
Therefore, it remains to show that $\ell(m)\geq \frac{c-c'}{2}$.
We first argue the following:
$\label{eqnLog}
 c-c'\underset{Obs.\ \ref{obs:spacingcheckpoints}}{\leq} \delta(1+\delta)^{k-2}\underset{\delta\leq 1}{\leq}\frac{(1+\delta)^{k-1}}{2}\underset{Def.\ of\ k}{<} \frac{i-c}{2}
.$
Since $P[m]$ was successfully checked against $c$ and since $m > \frac{i+c'}{2}$ we derive that $\ell(m) > \frac{i+c'}{2} - c$. Therefore,
$\ell(m) > \frac{i+c'}{2} - c =\frac{i-c}{2} + \frac{c'-c}{2}\underset{(\ref{eqnLog})}{>} c-c'  + \frac{c'-c}{2} = \frac{c-c'}{2}.$
\end{proof}\vspace{-1.5mm}
The following lemma shows that the checkpoints are sufficiently close in order to satisfy the multiplicative approximation. 
\begin{lemma}\label{lem:LogAccurateEstimation}
\aLog reports a midpoint
$m^*$ such that w.h.p. $\frac{\ell_{max}}{(1+\varepsilon)}\leq\tilde\ell(m^*)\leq\ell_{max}$.
\end{lemma}
\begin{proof}
In step \ref{logreport} of iteration $i$, \aLog reports the midpoint and length estimate of $P[m^*]$. 
We first argue that $\tilde\ell(m^*)\leq\ell(m^*)\leq\ell_{max}$. 
Let $i'$ be the last time $\tilde\ell(m^*)$ was updated by step \ref{item:mStarUpdate} of the algorithm.
By the condition of step \ref{item:mStarUpdate}, $S[c'+1,m^*]$ is the reverse of $S[m^*+1,i']$, where $c'=2\cdot m^*-i'$.
Hence, we derive $\ell(m^*)\geq i'-m = \tilde\ell(m^*) $. Moreover, by the definition of $\ell_{max}$ we have $\ell_{max}\geq\ell(m^*)$.\\
We now argue $\frac{\ell_{max}}{(1+\varepsilon)}\leq\tilde\ell(m^*)$.
Let $P[m_{max}]$ be a palindrome of maximum length, i.e., $\ell(m_{max})=\ell_{max}$.
Let $k$ be an integer such that $(1+\delta)^{k-1}<\ell_{max}\leq (1+\delta)^k$. 
Consider $\tilde\ell(m^*)$ after the algorithm processed $S[1,i']$, where $i'=m_{max}+(1+\delta)^{k-1}$.
By Observation \ref{obs:spacingcheckpoints}, there is a checkpoint in interval $[i'-2\cdot(1+\delta)^{k-3},i'-2\cdot(1+\delta)^{k-1}+\delta(1+\delta)^{k-3}]$. Let $c$ denote this checkpoint. \aLog successfully checked $P[m_{max}]$ against this checkpoint and therefore the value $\tilde\ell(m^*)$ is set to at least $m_{max}-c$.
We have $m_{max}-c\geq (1+\delta)^{k-1}-\delta(1+\delta)^{k-3}$. 
Thus,  
we have $\tilde\ell(m^*)\geq m_{max}-c\geq (1+\delta)^{k-1}-\delta(1+\delta)^{k-3}$. 
Therefore,
$\frac{\ell_{max}}{\tilde\ell(m^*)}\leq
\frac{(1+\delta)^k}{(1+\delta)^{k-1}-\delta(1+\delta)^{k-3}}
=\frac{(1+\delta)^3}{(1+\delta)^2-\delta}\leq(1+\delta)^2=1+\varepsilon$, 
where the last equation follows from 
$\delta=\sqrt{1+\varepsilon}-1$. 
\end{proof}
We are ready to prove Theorem \ref{the:log}. The correctness follows from Lemma \ref{lem:LogAccurateEstimation}. Lemma \ref{lem:numberOfCheckpoints} and Lemma \ref{lem:LcFormsarun} yield the claimed space.
In every iteration the algorithm processes every checkpoint in CL$_i$ in constant time. The number of checkpoints is bounded by Lemma \ref{lem:numberOfCheckpoints}.

\begin{proof}[Proof of Theorem \ref{the:log}]
In this proof, similar to Theorems \ref{the:firstpassisgood} and \ref{the:main},
we assume that the fingerprints do not fail w.h.p. as \aLog, similar to \A, does not take more than $n^2$ fingerprints during the processing of any input of length $n$.\\
\emph{Correctness:}
The correctness of the algorithm follows from Lemma \ref{lem:LogAccurateEstimation}.
It remains to argue that the space and time are not exceeded.
\newline
\emph{Space:}
The space required by \aLog is dominated by space needed to store the palindromes (corresponding midpoints and fingerprints) in $L_c$ for all $c\in$ CL$_i$.
Lemma \ref{lem:LcFormsarun} shows that for any $c\in$ CL$_i$ the list $L_c$ can be stored in constant space. Furthermore, Lemma \ref{lem:numberOfCheckpoints} shows that there are  $O\left(\frac{\log(n)}{\varepsilon\log(1+\varepsilon)}\right)$ elements in CL$_i$.\\
\emph{Running time:}
The running time is determined by step \ref{item:removeoldcheckpoints} and \ref{item:logncomparingstep}.
The algorithm goes in every iteration through all checkpoints in CL$_i$ which has
$O\left(\frac{\log(n)}{\varepsilon\log(1+\varepsilon)}\right)$ elements as Lemma \ref{lem:numberOfCheckpoints} shows.
For each checkpoint the steps (\ref{item:removeoldcheckpoints} and \ref{item:logncomparingstep})  take only constant time.
Thus, the required time to process the whole input is $O\left(\frac{n\log(n)}{\varepsilon\log(1+\varepsilon)}\right)$.
\end{proof}

%
\section{Lower bounds}\label{sec:LB}
\input{lowerbound.tex}
\bibliography{biblio}
\appendix
\newpage
%
{\center \Huge\bf  Appendix}
\section{Useful fingerprint properties}

\begin{lemma}\label{lem:combinefingerprints}
(Similar to Lemma 1 and Corollary 1 of \cite{Breslauer:11}) Consider two substrings $S[i,k]$ and $S[k+1,j]$ and their concatenated string $S[i,j]$  where $1\leq i\leq k\leq j\leq n$.
\begin{itemize}
\item $F^{F}(i,j) =\big( F^{F}(i,k)+r^{k-i+1}\cdot F^{F}(k+1,j)\big)\text{ mod }p.$
\item $F^{F}(k+1,j) = r^{-(k-i+1)}\big( F^{F}(i,j)-F^{F}(i,k)\big)\text{ mod }p.     $
\item$F^{F}(i,k) = \big(F^{F}(i,j)-r^{k-i+1}\cdot F^{F}(k+1,j)\big)\text{ mod }p. $
\end{itemize}

\end{lemma}
The authors of \cite{Breslauer:11} show that, for appropriate choices
of $p$ and $r$, it is very unlikely that two different strings share the same fingerprint.

\section{Complementary Palindromes}
\begin{definition}\label{complpalindrome}
Let $f:\Sigma \rightarrow \Sigma$ be a function indicating a complement for each symbol in $\Sigma$.
A string $S\in\Sigma^n$ with length $n$ contains a \emph{complementary palindrome} of length $\ell$ with midpoint 
$m\in\{\ell\kdots n-\ell\}$ if $S[m-i+1]=f(S[m+i])$ for all $i\in\{1\kdots\ell\}$.
\end{definition}
The fingerprints can also be used for finding complementary palindromes:
If one changes the forward \mfp to be $F_c^F(1,l) = \left( \sum_{i=1}^l{f(S[i]) \cdot r^i} \right) \mod\ p $ (as opposed to $F^F(1,l) = \left( \sum_{i=1}^l{S[i] \cdot r^i} \right) \mod\ p $) in all algorithms in this paper, then we obtain the following observation. 
\begin{observation}
All algorithms in this paper can be adjusted to recognize complementary palindromes with the same space and time complexity. 
\end{observation}

\begin{proof}[Proof of Lemma \ref{lem:invariantsgood}]
Fix an arbitrary palindrome $P[m]$. First we assume $ \ell(m) < \sqrt n$. Then \sA reports $m$ and $\tilde\ell(m)$ in
step \ref{alg:onlyoutputlongpalindromes} of iteration $m+\sqrt n$ which is the iteration where the entire palindrome is in the sliding window. 
Furthermore, in the step \ref{alg:insideslidingwindow}, 
the algorithm checks for all 
$j\in\{\floor{\varepsilon\sqrt n},2\cdot\floor{\varepsilon\sqrt n},\kdots r\cdot\floor{\varepsilon\sqrt n},\sqrt n\}$, where $r$ is the maximum integer s.t. $r\cdot\floor{\varepsilon\sqrt n}<\sqrt n$,
if $F^R(m - j+1 ,m)=F^F(m +1, m+j )$,
then sets $\tilde\ell(m,i)=j$.
Let $j_m$ be the maximum $j$ such that $F^R(m - j+1 ,m)=F^F(m +1, m+j)$
is satisfied at iteration $m+\sqrt n$. 
Then $P[m]$ covers the \slidingpair with length $j_m$ but does not cover the \slidingpair with length $j_m + \floor{\varepsilon \sqrt n}$. Since \sA sets $\tilde\ell(m)=j_m$ we have $\ell(m) - \varepsilon \sqrt n < \tilde\ell(m) \leq \ell(m)$.\\
Now we assume $\ell(m) \geq \sqrt n$. Step \ref{enu:NewPalindromeInSlidingWindow} of iteration $m+\sqrt n$ adds $R_S(m,i)$ to $L_{i-1}$. 
We show for every $i\geq m+\sqrt n$ that the following holds: After \sA read $S[1,i]$ (iteration $i$) we have  $\ell(m,i)-\varepsilon \sqrt{n} < \tilde\ell(m,i) \leq \ell(m,i)$. We first show the first inequality and afterwards the second.
Define $i'\leq i$ to be the last iteration where the algorithm updated $\tilde\ell(m,i')$ in Step \ref{enu:CompareWithMidpoints}, i.e., it sets $\tilde\ell(m,i')=\ell(m,i')=i'-m$.
\begin{itemize}
\item $\ell(m,i)-\varepsilon \sqrt{n} < \tilde\ell(m,i)$: 
We first show $\ell(m,i)<i'+\varepsilon \sqrt n - m$ by distinguishing between two cases:
\begin{enumerate}
\item $ i' > i - \varepsilon \sqrt n:$  By definition of $\ell(m,i)$, we have $\ell(m,i)\leq i-m$. And thus $\ell(m,i) < i'+\varepsilon \sqrt n - m$.
\item $ i' \leq i - \varepsilon \sqrt n :$ Since the estimate of $m$ was updated at iteration $i'$
we know that there is a checkpoint at index $2m-i'$ and therefore we know that due to step \ref{checkpointing} there is a checkpoint at index $2m- (i'+\floor{\varepsilon \sqrt n})$.
Since $i'$ is the last iteration where the estimate of $P[m]$ was updated we infer that $S[2m- (i'+\floor{\varepsilon \sqrt n}) ,m]$ was not the reverse of $S[m+1,i'+\floor{\varepsilon \sqrt n}]$.
Hence, $\ell(m,i) < i' + \floor{\varepsilon \sqrt n} - m \leq i' + \varepsilon \sqrt n - m$.
\end{enumerate}
With $\ell(m,i)< i'+\varepsilon \sqrt n-m$ we have 
$$\ell(m,i)< i'+\varepsilon \sqrt n-m=\ell(m,i')+\varepsilon \sqrt n=\tilde\ell(m,i')+\varepsilon \sqrt n=\tilde\ell(m,i)+\varepsilon \sqrt n.$$ The last equality holds since $i'$ was the last index where $\tilde\ell(m,i)$ was updated.

\item $\tilde\ell(m,i) \leq \ell(m,i)$:
Whenever $\tilde\ell(m,i)$ is updated to $\ell(m,i')$ by the algorithm, this means that $F^F(m+1,i')=F^R(2m-i',m)$ and since we assume that fingerprints do not fail we have that $S[m-\tilde\ell(m,i')+1,m]$ is the reverse of $S[m+1,m+\tilde\ell(m,i')]$.
It follows that $\tilde\ell(m,i)=\tilde\ell(m,i')$ and $\ell(m,i) \geq \ell(m,i')$. 
\end{itemize}
Furthermore, step \ref{sALaststep} reports $L_n$ at iteration $n$ which includes $m$ and $\tilde\ell(m)$.
\end{proof}
\sA requires linear space in the worst-case. 
The algorithm stores a list of all $\sqrt{n}$-palindromes
For an $S$, containing a linear number of $\sqrt n$-palindromes. 
\sA requires linear worst-case space. As an example for such a stream consider
$S=a^{n}$ with $a\in \Sigma$. All indices in the interval $[\sqrt n,n-\sqrt n]$ are midpoints of  $\sqrt{n}$-palindromes.
%
%
%

\section{\A}

\begin{proof}[Proof of Corollary \ref{cor:runhasstructure}]
We prove this by induction. Suppose $j_0$ is the highest index where $m_{j_0}<m_1+\ell^*$. By Definition \ref{def:compressible}, we have $j_0 \geq 3$. We start by proving the induction basis.
By Lemma \ref{lem:midpointshavestructure}, the claim holds for $m_1,m_2\kdots m_{j_0}$, i.e., they are equally spaced and $S[m_1+1,m_{j_0}]$ is a prefix of $ww^R\dots ww^R$. For the inductive step we assume that the claim holds for $m_{j-1}$. Consider the midpoints $m_{j-1},m_j,m_{j+1}$. Since $m_{j+1}-m_{j-1}\leq \ell^*$,  Lemma \ref{lem:midpointshavestructure} shows that those midpoints fulfill the claimed structure.
\end{proof}

\begin{proof}[Proof of Lemma \ref{lem:palindromesOnOneSideOfRun}]
We prove the first case where $m_j$ is in the first half, i.e., $j<{\frac{h+1}{2}}$. The other case is similar.
By Corollary \ref{cor:runhasstructure}, $S[m_1,m_h]$ is of the form $ww^Rww^R...$ and  $m_j=m_1+(j-1)\cdot |w|$, where $w$ is $S[m_1+1,m_2]$. Define $m_0$ to be the index $m_1-|w|$. Since $\ell(m_1,i) \geq \ell^* \geq |w|$ we have $S[m_0+1,m_1]=w^R$.\\
By Corollary \ref{cor:runhasstructure}, we have that $S[m_0+1,m_j]^R = S[m_j+1,m_{2j}]$. This implies $\ell(m_j,i)\geq j\cdot |w|$. Define $k$ to be $\ell(m_j,i) - j\cdot |w|$. We show that $k=\ell(m_0,i)$. By definition, $k$ is the length of the longest suffix of $S[1,m_0]$ which is the reverse of the prefix of $S[m_{2j}+1,n]$. Corollary \ref{cor:runhasstructure} shows that $S[m_{2j}+1,m_{2j}+\ell^*]$ is equal to $S[m_{0}+1,m_{0}+\ell^*]$ as both are prefixes of $w^Rww^R...$. 
Therefore, $k$ is also the same as the length of the longest suffix of $S[1,m_0]$ which is the reverse of the prefix of $S[m_{0}+1,m_{0}+\ell^*]$, i.e., $k=max \{ k' | S[m_0 -k' + 1,m_0]^R= S[m_{0}+1,m_{0}+k']\}=\ell(m_0,i)$. Thus, $\ell(m_j,i)=\ell(m_0,i)+j\cdot |w|$.
\end{proof}

\begin{proof}[Proof of Lemma \ref{lem:runconstantsize}]
Fix an index $j$. We prove that we can retrieve $R_S(m_j,i)$ out of the \finished representation.
Corollary \ref{cor:runhasstructure} gives a formula to retrieve $m_j$ from the corresponding \finished.
Formally, $m_j=m_1+(j-1)\cdot(m_1-m_2)$.\\
Corollary  \ref{cor:runhasstructure} shows that $S[m_1 +1, m_h]$ follows the structure $ww^Rw^R\cdots$ where $w=S[m_1+1,m_2]$.\\ This structure allows us to retrieve $F^F(1,m_j),F^R(1,m_j)$, since we have
$F^F(1,m_j)=$\\ $\phi^F(S[1,m_1]ww^Rww^R\cdots)$.\\
We know argue that the length estimates have the same accuracy as \singles. Note that the proof of Lemma \ref{lem:invariantsgood} shows that after iteration $i$ and any $R_S(m,i)$ we have $\ell(m,i)-\varepsilon \sqrt{n} < \tilde\ell(m,i) \leq \ell(m,i)$.
We show that one can retrieve the length estimate for palindromes which are not stored explicitly by using the following equation. The equation is motivated by Lemma \ref{lem:palindromesOnOneSideOfRun}.
$$\tilde\ell(m_j,i)=\begin{cases}\tilde\ell(m_1,i)+(j-1)\cdot(m_2-m_1)&j<\frac{h+1}{2}\\ \tilde\ell(m_h,i)+(h-j)\cdot (m_2-m_{1}) & j>\frac{h+1}{2}\end{cases}.$$
Let $i'$ be the index where $R_F$ was finished.
We distinguish among three cases: 

\begin{enumerate}
\item $m_j=m_1$: Since the algorithm treats $P[m_j]$ as an \single in terms of comparisons, $\ell(m_j,i)-\varepsilon \sqrt{n} < \tilde\ell(m_j,i) \leq \ell(m_j,i)$ holds.   
\item $m_j \in \{m_{\lfloor \sfrac{(h+1)}{2}\rfloor}, m_{\lceil  \sfrac{(h+1)}{2} \rceil}, m_h \}$: At index $i'$ \A executes step \ref{enu:AlgoFinishesRun} and one can verify that $\ell(m_j,i')-\varepsilon \sqrt{n} < \tilde\ell(m_j,i') \leq \ell(m_j,i')$ holds. For all $i \geq i'$ palindrome $P[m_j]$ is treated as an \single in terms of comparisons. Thus, $\ell(m_j,i)-\varepsilon \sqrt{n} < \tilde\ell(m_j,i) \leq \ell(m_j,i)$ holds.
\item Otherwise we assume \withoutlog $1<j< \floor{\frac{h+1}{2}}$. 
Lemma \ref{lem:palindromesOnOneSideOfRun} shows that $\ell(m_j,i)-\tilde\ell(m_j,i)=\ell(m_1,i)-\tilde\ell(m_1,i)$.
We know
$0\leq\ell(m_1,i)-\tilde\ell(m_1,i)<\varepsilon\sqrt{n}$. 
Thus, $\ell(m_j,i)-\varepsilon \sqrt{n} < \tilde\ell(m_j,i) \leq \ell(m_j,i)$. 
\end{enumerate} 
\end{proof} 
\begin{proof}[Proof of Theorem \ref{the:firstpassisgood}]
Similar to other proofs in this paper we assume that fingerprints do not fail as we take less than $n^2$ fingerprints and by Lemma \ref{lem:breslauer2}, the probability that a fingerprint fails is at most $1/(n^4)$.
Thus, by applying the union bound the probability that no fingerprint fails is at least $1-n^{-2}$.\\
\emph{Correctness:} Fix an arbitrary palindrome $P[m]$.
For the case $\ell(m) < \sqrt n$ there is no difference between \sA and \A, so the correctness follows from Lemma \ref{lem:invariantsgood}.
In the following, we assume $\ell(m) \geq \sqrt n$.
Firstly, we argue that  $R_S(m,n)$ is stored in $\hat L_n$. 
At index $i=m+\sqrt n$, \A adds $P[m]$ to $\hat L_i$. 
The algorithm does this by using the longest sliding fingerprint pair which guarantees that 
if  $S[i-2\sqrt n + 1, i-\sqrt n ]$ is the reverse of $S[i-\sqrt n + 1, i ]$, then the fingerprints of sliding window are equal.
Moreover, a palindrome is never removed from $\hat L_i$ for $1\leq i \leq n$. Additionally, Lemma \ref{lem:runconstantsize} shows how to retrieve the midpoint. Hence,  $R_S(m,n)$ is stored in $\hat L_n$.\\
We now argue $\ell(m)-\varepsilon \sqrt{n} < \tilde\ell(m) \leq \ell(m)$. Palindromes are stored in an \single, \finished and \current. 
Since we are only interested in the estimate $\tilde\ell(m)$ after the $n^{th}$ iteration of \sA and since
the algorithm finishes an \current at iteration $n$, we know that there are no \currents at after iteration $n$.

\begin{enumerate}
\item  $R_S(m,n)$ is stored as an \single. Since \singles are treated in the same way as in \sA,  $\ell(m)-\varepsilon \sqrt{n} < \tilde\ell(m) \leq \ell(m)$ holds by Lemma \ref{lem:invariantsgood}.

\item  $R_S(m,n)$ is stored in the \finished $R_F$. Then Lemma \ref{lem:runconstantsize} shows the correctness.
\end{enumerate}
Furthermore, the algorithm reports $\hat L_n$ step \ref{sALaststep} of iteration $n$.
\newline
\emph{Space:} 
The number of checkpoints equals $\floor{\sfrac{n}{\lfloor\varepsilon\sqrt n\rfloor}} \leq \sfrac{2n}{\varepsilon\sqrt n}= O( \sfrac{ \sqrt n}{\varepsilon})$, since $\varepsilon \geq \sfrac{1}{\sqrt n}$. Moreover,
there are $O(\sfrac{n}{\varepsilon})$ \slidingpairs which can be stored in $O(\sfrac{n}{\varepsilon})$ space.
The sliding window requires $2\sqrt n$ space.
The space required to store the information of all $\sqrt n$-palindromes is bounded by $O(\sqrt n)$:
By Observation \ref{obs:constantNumberOfCandi}, the number of \singles and \comps in an interval of length $\sqrt n$ is bounded by a constant.
Each \comp and each \single can be stored in constant space.
Thus, in any interval of length $\sqrt n$ we only need constant space and thus altogether $O(\sqrt n)$ space for storing the information of palindromes.
\newline
\emph{Running time:} 
The running time of the algorithm is determined by the number of comparisons done at lines \ref{alg:comparison} and \ref{alg:insideslidingwindow}.
First we bound the number of comparisons corresponding to line \ref{alg:comparison}.
For all $\sqrt n$-palindromes we bound the total number of comparisons by $O(\frac{n}{\varepsilon})$:
The \A checks only explicitly stored palindromes with checkpoints and therefore with \singles and at most 4 midpoints per \comp.
As shown in Observation \ref{obs:constantNumberOfCandi}, there is at most a constant number of explicitly stored midpoints in every interval of length $\sqrt n$.
In total, we have $O(\sqrt n)$ explicitly stored midpoints and $O(\sfrac{\sqrt{n}}{\varepsilon}) $ fingerprints of checkpoints. 
We only check each palindrome at most once with each checkpoint
\footnote{We can use an additional queue to store the index where the algorithm needs to \emph{check} a checkpoint with a palindrome. }.
Hence, the total number of comparisons is in order of $ O(\sfrac{n}{\varepsilon} )$. 
Now, we bound the running time corresponding to Step \ref{alg:insideslidingwindow}.
This step has two functions: There are $O(\sfrac{1}{\varepsilon})$  
\slidingpairs which are updated every iteration. This takes $O(\sfrac{1}{\varepsilon})$ time. Additionally, the middle of the sliding window is checked with at most  $O(\sfrac{1}{\varepsilon})$ \slidingpairs. Thus, the time for Step \ref{alg:insideslidingwindow} of the algorithm is bounded by $O(\sfrac{n}{\varepsilon})$. 
\end{proof}
%
%
\section{\algB}\label{analysingsecondpass}
The analysis of \algB is based on the observation that, after removing palindromes which are definitely shorter then the longest palindrome, at any time the number of stored uncertain intervals is bounded by a constant.
The following Lemma shows that only the palindromes in the middle are strictly longer than the other palindromes of the run. This allows us to remove all palindromes which are not in the middle of the run. The techniques used in the lemma are very similar to the ideas used in Lemma \ref{lem:palindromesOnOneSideOfRun}. Let $\hat L_n$ be the list after the first pass. 
\begin{lemma}\label{lem:maxinmid}
Let $m_1,m_2,m_3,...,m_h$ be midpoints of a maximal $\ell^*$-run in $S$.
For every $j\in \{1\kdots h\}\setminus\{\floor{(h+1)/2},\ceil{(h+1)/2}\}$, $$\ell(m_j)< max\{\ell(m_{\floor{(h+1)/2}}),\ell(m_{\ceil{(h+1)/2}})\}.$$
\end{lemma}
\begin{proof}
If $h$ is even, the claim follows from Lemma \ref{lem:palindromesOnOneSideOfRun}.
Therefore, we assume $h=2d-1$ which means that $m_{\floor{(h+1)/2}}=m_{\ceil{(h+1)/2}}=m_d$.
Hence, we have to show that  $\ell(m_j) < \ell(m_d) $.
Define $w$ exactly as it is defined in Lemma \ref{lem:palindromesOnOneSideOfRun} to be $S[m_1 + 1, m_2]$.
Note that $S[m_{h-1} + 1, m_h]=w^R$.
We need two claims:
\begin{enumerate}
\item $\ell(m_d)\geq (d-1)|w|+ min\{ \ell(m_1), \ell(m_h) \} $\\
Proof: Suppose $\ell(m_d) < (d-1)|w|+ min\{ \ell(m_1), \ell(m_h) \} $.
We know that $S[m_1 + 1, m_d]$ is the reverse of $S[m_d + 1, m_h]$. 
Therefore, $\ell(m_d) =(d-1)|w| + k$ where $k$ is the length of the longest suffix of  $S[1,m_1]$ which is the reverse of the prefix of $S[m_h+1,n]$.
Formally, $k=max \{ k' | S[m_1 -k' + 1,m_1]^R= S[m_{h}+1,m_{h}+k']\}$.\\
Define $\ell'\triangleq min\{ \ell(m_1), \ell(m_h) \}$.
Thus, it suffices to show that $k \geq min\{ \ell(m_1), \ell(m_h) \} = \ell'$.
Observe, $\ell'<\ell^*+ |w|$ since otherwise $m_1 - |w|$ or $m_h + |w|$ would be a part of the run.
Since $m_1$ is the midpoint of a palindrome of length $\ell'<\ell^*+ |w|$, by Corollary \ref{cor:runhasstructure}, the left side of $m_1$, i.e., 
 $S[m_1 - \ell' + 1, m_1]$ is a suffix of length $\ell'$ of $ww^R\dots ww^R$. Similarly, $S[m_h+1,m_h + \ell']$ is a prefix of length $\ell'$ of $ww^R\dots ww^R$.
These two facts imply that  $S[m_1-\ell' +1, m_1]$ is the reverse of $S[m_h+1,m_h + \ell']$ and thus $k \geq \ell'$.
\QEDB
\item $|\ell(m_{h}) - \ell(m_{1}) |< |w|$ \\
Proof: \withoutlog, let $ \ell(m_h) \geq  \ell(m_1)$, then suppose $\ell(m_h)-\ell(m_1) \geq |w| $.
This implies that $\ell(m_h) \geq |w| + \ell(m_1) \geq |w| + \ell^*$.
By Lemma \ref{cor:runhasstructure}, we derive  that the run is not maximal, i.e., there is a midpoint of a palindrome with length of $\ell^*$ at index $m_h + |w|$.
A contradiction.
\QEDB
\end{enumerate}
Using these properties we claim that  $\ell(m_d) > \ell(m_j)$:
\begin{enumerate}
\item For $j < d:$ $\ell(m_d)\underset{Property\ 1}{\geq}(d-1)|w|+ min\{ \ell(m_1), \ell(m_h) \} \underset{Property\ 2}{>} (d-1)|w|+  \ell(m_1) - |w| =  (d-2)|w|+  \ell(m_1) - |w| \geq  (j-1)|w|+  \ell(m_1) \underset{Lemma\ \ref{lem:palindromesOnOneSideOfRun}}{=}\ell(m_j)$.
\item For $j > d:$ $\ell(m_d)\underset{Property\ 1}{\geq}(d-1)|w|+ min\{ \ell(m_1), \ell(m_h) \} \underset{Property\ 2}{>} (d-1)|w|+  \ell(m_h) - |w| =  (d-2)|w|+  \ell(m_h) - |w| \geq  (j-1)|w|+  \ell(m_h) \underset{Lemma\ \ref{lem:palindromesOnOneSideOfRun}.}{=}\ell(m_j)$.
\end{enumerate}
This yields the claim.
\end{proof}
We conclude this section by proving Theorem \ref{the:main} covering the correctness of the algorithm as well as the claimed space and time bounds. 
We say a palindrome with midpoint $m$ \emph{covers} an index $i$ if $|m - i| \leq \ell(m)$.
\begin{proof}[Proof of Theorem \ref{the:main}] 
In this proof, similar to Theorem \ref{the:firstpassisgood},
we assume that the fingerprints do not fail w.h.p. .\\
For the case $\ell_{max} < \sqrt n$ it is easy to see that the algorithm satisfies the theorem.\\
Therefore, we assume $\ell_{max} \geq \sqrt n$.\\
\emph{Correctness:} 
After the first pass we know that due to Theorem \ref{the:firstpassisgood} all $\sqrt n$-palindromes are in $L$.
The algorithm removes some of those palindromes and we argue that a palindrome which is removed from $L$ cannot be the longest palindrome. A palindrome removed in
\begin{itemize}
\item Step \ref{alg:onlykeepthemiddle} is, by Lemma \ref{lem:maxinmid},  strictly shorter than the palindromes of the middle of the run from which it was removed.
\item Step \ref{alg:deleteshortones} with midpoint $m$ has a length which is bounded by
$\tilde\ell_{max} - \tilde\ell(m) \geq \sfrac{ \sqrt n}{2}$.
We derive $\ell_{max}\geq \tilde\ell_{max} \geq \sfrac{ \sqrt n}{2}  + \tilde\ell(m) >\ell(m)$ where the last inequality follows from $\ell(m) - \sfrac{ \sqrt n}{2}  < \tilde\ell(m) \leq \ell(m)$. 
\end{itemize} 
Therefore, all longest palindromes are in $L^*$.
Furthermore, the exact length
of $P[m]$ is determined at iteration $m+\ell(m)+1$ since this is the first iteration where $S[m-(\ell(m)+1)+1]\not=S[m+\ell(m)+1]$.
In Step \ref{alg:lengthdetermined}, the algorithm sets the exact length $\ell(m)$.
If $\ell(m)=\ell_{max}$,
then the algorithm reports $m$ and $\ell_{max}$ in step \ref{alg:laststepofexact} of iteration $n$.\\
\emph{Space:}
For every palindrome we have to store at most one uncertain interval.
At iteration $i$, the number of uncertain intervals we need to store equals the number of palindromes which cover index $i$. We prove in the following that this is bounded by 4. We assume $ \varepsilon$ to be $\sfrac{1}{2}$ and $\ell_{max} \geq \sqrt n$. 
Define $\tilde\ell_{min}$ to be the length of the palindrome in $L^*$ for which the estimate is minimal. All palindromes in $L^*$ have a length of at least $\sqrt n$, thus $\tilde\ell_{min}\geq \sqrt n$.
We define the following intervals: 
\begin{itemize}
\item$\mathcal{I}_1=(i-\tilde\ell_{min}- \sqrt n, i-\tilde\ell_{min} ]$
\item$\mathcal{I}_2=(i-\tilde\ell_{min}, i]$
\end{itemize}
Recall that the algorithm removes all palindromes which have a length of at most $\tilde\ell_{max}-\sfrac{\sqrt n}{2}$ and thus $\tilde\ell_{min} \geq \tilde\ell_{max}-\sfrac{\sqrt n}{2}$.
Additionally, we know by Theorem \ref{the:firstpassisgood} that $\ell_{max}-\tilde\ell_{max} < \sfrac{\sqrt n}{2}$. We derive: $\ell_{max}-\tilde\ell_{min} \leq \ell_{max}-\tilde\ell_{max}+\sfrac{\sqrt n}{2} < \sfrac{\sqrt n}{2}  + \sfrac{\sqrt n}{2}   =\sqrt n$ and therefore
$\ell_{max}< \sqrt n+\tilde\ell_{min}$.
Hence, there is no palindrome which covers $i$ and has a midpoint outside of the intervals $\mathcal{I}_1$ and $\mathcal{I}_2$.
It remains to argue that the number palindromes which are centered in $\mathcal{I}_1$ and $\mathcal{I}_2$ and stored in $L^*$ is bounded by four: 
Suppose there were at least four palindromes in $\mathcal{I}_1$ and $\mathcal{I}_2$ which cover $i$.
Lemma \ref{lem:midpointshavestructure} shows that in any interval of length $\tilde\ell_{min}$ either the number of palindromes is bounded by two or they form an $\tilde\ell_{min}$-run. Thus there has to be at least one $\tilde\ell_{min}$-run.
Recall that step \ref{alg:onlykeepthemiddle} keeps for all \finisheds only the midpoints in the middle of the run. 
The first pass does not create \finisheds for all runs where the difference between two consecutive midpoints is more than $\sfrac{\sqrt n}{2}$ (See Definition \ref{defMaximalRun} and step \ref{stepaddtorun} of \A). 
Thus, there is a run where the distance between two consecutive midpoints is greater than $\sfrac{\sqrt n}{2}$.
By Lemma \ref{lem:palindromesOnOneSideOfRun}, the difference between $\ell(m)$ for distinct midpoints $m$ of one side of this run is greater than $\sfrac{\sqrt n}{2}$. Since the checkpoints are equally spaced with consecutive distance of $\sfrac{\sqrt n}{2}$, the difference between $\tilde\ell(m)$ for distinct midpoints $m$ of one side of this run is greater than $\sfrac{\sqrt n}{2}$ as well.
This means that just the palindrome(s) in the middle of the run, say $P[m]$, satisfy the constraint $\tilde\ell(m) \geq \tilde\ell_{max} - \sfrac{ \sqrt n}{2}$ and therefore the rest of the palindromes of this run would have been deleted. A contradiction.
Thus, there are at most two palindromes in both intervals and thus in total at most four palindromes and four uncertain intervals. This yields the space bound of $O(\sqrt n)$.
\\
\emph{Running time:}
As shown in Theorem \ref{the:firstpassisgood} and Lemma \ref{lem:smallellmax} the running time is in  $O(n)$.
The preprocessing and the last step of the second pass can be done in $O(n)$.
If the first pass returned an $\ell_{max} < \sqrt n$ then the running time of the second pass is trivially in $O(n)$.
Suppose $\ell_{max} \geq \sqrt n$. The preprocessing and the last step of \algB can be done in $O(n)$ time and they are only executed once.
The remaining operations can be done in constant time per symbol.
This yields the time bound of $O(n)$.
\end{proof}

%
%
%
%

\section{Variants of \algA}
In this section we present two variants of \A. 
The first variant is similar to \A, but instead of reporting all palindromes it reports a palindromes  $P[m]$ iff $ \ell(m) > \ell_{max} - \varepsilon \sqrt n$ assuming that $\ell_{max} \geq \sqrt n$.
If the algorithm is run on an input where $\ell_{max} < \sqrt n$, then the algorithm realizes this and does not report any palindromes. It require $\omega(\frac{\sqrt n}{\varepsilon})$ space to output all palindromes of size $\ell_{max}$ if  $\ell_{max} < \sqrt n$.
The described variant can be implemented in the following way:
Run \A and before returning the final list $L_n$ trim the list $L_n$ by removing all short palindromes. For details see the preprocessing of \algB (which is introduced in Section \ref{section:algB}).
This leads to Observation \ref{ObsVariant}.\\
The second variant reports one of the longest palindromes and the precise $\ell_{max}$ if  $\ell_{max}< \sqrt n$. In case $\ell_{max} \geq \sqrt n$, the algorithm detects this, but does not report the precise $\ell_{max}$.
The reason is that we can store a small palindrome in our memory and it is not possible to report $\ell_{max} (\geq \sqrt n)$ in space $o(n)$.
%
%
In following we present the proof of Lemma \ref{lem:smallellmax}.
\begin{proof}[Proof of Lemma \ref{lem:smallellmax}] 
The algorithm uses a sliding window of size $2\sqrt n$. At iteration $i$, the middle of the sliding window is $m=i -\sqrt n$. Let $\ell_{i-1} < \sqrt n$ be the length of the longest palindrome at index $i-1$. The algorithm does the following.\\
%
\emph{Algorithm:}
Initialize $k=0$.
At iteration $i$ compare $F^R(m -k + 1, m)$ with $F^F(m+1, m+k)$ for $k=\ell_{i-1}+1$ to $\sqrt n$ until they are unequal: Set $\ell_i$ to be the highest value for $k$ such that $F^R(m -k + 1, m)$ equals $F^F(m+1, m+k)$.
If $\ell_i > \ell_{i-1}$, then set $longest = P[m]$. 
If $\ell_i $ is smaller than  $\ell_{i-1}$, set $\ell_{i}= \ell_{i-1}$. After iteration $n$ output $\ell_{max}$ and $longest$.\\
\emph{Correctness:} For $k\in \mathbb{N}$ if $S[i-k+1,i]$
is the reverse of $S[i+1,i+k]$, then the fingerprints are equal.
Let $m$ be the first palindrome in $S$ with $\ell(m) = \ell_{max}$, then at iteration $m+\sqrt n$ the algorithm compares the fingerprint $F^R(i-(\ell_{max}+1)+1, i)$
with $F^F(i+1, i+\ell_{max}+1)$ and sets $\ell_m=\ell_{max}$.
The value of $\ell_m$ is not changed afterwards.\\
\emph{Space:} The algorithm stores two fingerprints, the longest palindrome found so far, and the sliding window of size $2\sqrt n$ which results in $O( \sqrt n )$ space. \\
\emph{Running time:} Sliding the two fingerprints takes $O(n)$ time in total.
The fingerprints are extended at most $\sqrt n$ times.
Storing the longest palindrome takes at most $O(\sqrt n)$ time and this is done at most $\sqrt n$ times. 
\end{proof}

\end{document}

%% file: lowerbound.tex
In this section, we prove lower bounds on the space of  any algorithm finding the longest palindrome with probability $1$.
We give in Lemma \ref{lem:detLB} a probability distribution over input streams.
This distribution gives a lower bound on the required space of an optimal deterministic algorithm which approximates the length of the longest palindrome within an additive error.\\
Theorem \ref{thm:additiveLB} applies Yao's principle (for further reading see \cite{Motwani:97}) and shows a lower bound on the space of any randomized algorithm which approximates the length of the longest palindrome.
Recall that $\Sigma$ denotes the set of all input symbols.
A \emph{memory cell} is a memory unit to store a symbol from $\Sigma$.
Let $C(S,A)$ denote the required space of algorithm $A$ on the input stream $S$.
\begin{lemma}\label{lem:detLB}
Let $A$ be an arbitrary deterministic algorithm which approximates the length of the longest palindrome up to an additive error of $e_r$ elements with probability 1.
For any positive integers $m$ and $e_r$, there is a probability distribution $p$ over $\Sigma^{2m(2e_r+1)+4e_r}$ such that 
$E[C(S,A)] \geq m$, where $S\in \Sigma^{2m(2e_r+1)+4e_r}$ is a random variable and follows the distribution $p$.
\end{lemma}
\begin{proof}
We construct a set of input streams $\mathcal{S}$ and we discuss afterwards the expected space that $A$ needs to process a random input from $S\in \mathcal{S}$. 
We define the string $[2e_r]=1,2,3,...,2\cdot e_r$ and similarly $[2e_r]^R=2 \cdot  e_r,2 \cdot e_r - 1\kdots 2,1$.
Define $\mathcal{S}$ such that each element of $\mathcal{S}$ consists of $2m$ arbitrary symbols from $\Sigma$ separated by the string $[2e_r]$ in the first half and by $[2e_r]^R$ in the second half. 
Formally, $$ S=\{[2e_r]a_1[2e_r]a_2[2e_r]...a_m[2e_r][2e_r]^R a_{m+1}[2e_r]^R a_{m+2}\dots[2e_r]^R a_{2m}[2e_r]^R \ \mid \ \forall j, a_j\in\Sigma \}.$$
We define the input distribution $p$ to be the uniform distribution on $\mathcal{S}$. 
Let $\ell$ denote the length of a stream $S\in \mathcal{S}$, i.e., $\ell=2m(2e_r +1) + 4e_r$.
In the following, we show that $A$ needs to store $a_1 \kdots a_m$ in order to have an approximation of at most $e_r$.
\\
If $A$ does not store all of them, then $A$ behaves the same for two streams $S_1,S_2 \in \mathcal{S}$ 
with 
\begin{enumerate}
\item $S_1[1,\sfrac{\ell}{2}] \not= S_2[1,\sfrac{\ell}{2}]$
\item $S_1[1,\sfrac{\ell}{2}]=S_2[\sfrac{\ell}{2},\ell]^R =S_1[\sfrac{\ell}{2},\ell]^R $
\end{enumerate}
Let $x=j(2e_r+1)$ be the highest index of $S_1[1,\sfrac{\ell}{2}]$ such that $S_1[x] \not= S_2[x]$ for $j\in \{ 1\kdots m\}$.
If $A$  has the same memory content at index $\sfrac{\ell}{2}$ for $S_1$ and $S_2$, then $A$ returns the same approximation for both streams, but their actual longest palindromes are of lengths  $(m+1)2e_r$ and $(m-j+1)\cdot 2e_r$ where $j \geq 1$. Hence, no matter which approximation $A$ returns, it differs by more than $e_r$ of either $S_1$ or $S_2$.
Thus, $A$ must have distinct memory contents after reading index $\sfrac{\ell}{2}$.\\
In what follows we argue that the expected required space to store $a_1 \kdots a_m$ is $m$.\\
We use Shannon's entropy theorem (for further reading on information theory see \cite{shannon:48}) to derive a lower bound on the expected size of the memory. 
By Shannon's entropy theorem, the expected length of this encoding cannot be smaller than the distribution's entropy. 
Fix an arbitrary assignment for the variables $a_1 \kdots a_m$.  The probability for a string to $S$ to have the same assignment is $\frac{1}{|\Sigma|^m}$.
The entropy of an uniform distribution is logarithmic in the size of the domain. Hence, $log(|\Sigma|^m)=m\cdot log(	|\Sigma|)$ (or $m$ memory cells) is the lower bound on the expected space of $A$.
\end{proof}
Theorem \ref{thm:additiveLB} uses Yao's technique to prove the lower bound for randomized algorithms' space on the worst-case input, using deterministic algorithms' space on random inputs.
\begin{theorem}\label{thm:additiveLB}
Any randomized algorithm for approximating the length of the longest palindrome in a stream of length $\ell$ has the following property: In order to approximate the length up to an \emph{additive error} of $e_r$ elements with probability $1$ it must use $\Omega(\sfrac{\ell}{e_r})$ expected space.
\end{theorem}
\begin{proof}
Let $\mathcal{S}$ be the set of all possible input streams. Let $\mathcal{D}$ be the set of all deterministic algorithms.  By Yao's principle (\cite{Motwani:97}), we derive the following.
For any random variable $S_p$ over input streams which follows a probability distribution $p$ and  for any randomized algorithm $D_q$ which is a probability distribution $q$ on deterministic algorithms we have:
$$ \underset{A\in \mathcal{D}}{Min} E[C(s_p,A)]\leq \underset{I\in \mathcal{S}}{Max} E[C(I,D_q)]$$
Lemma \ref{lem:detLB} gives a lower bound for the left hand side of the above inequality: It shows that for a stream of length $2m(2e_r+1)+4e_r$ at least $m$ memory cells are required in order to achieve an additive error of at most $e_r$. Therefore, the required space for $\ell=2m(2e_r+1)+2e_r$ is $\Omega(\sfrac{\ell}{e_r})$. One can generalize this for any $\ell$ by using padding.
Thus, 
$$ \Omega(\sfrac{l}{e})\leq \underset{I\in \mathcal{S}}{Max} E[C(I,D_q)].$$
\end{proof}
%
%
Corollary \ref{cor:preciseAdditiveLB} is another direct implication that can be obtained by setting $e_r=1$.
\begin{corollary}\label{cor:preciseAdditiveLB}
There is no any randomized algorithm that computes the length of the longest palindrome in a given string precisely with probability $1$ and uses a sublinear number of memory cells.
\end{corollary}